\documentclass[journal]{IEEEtran}
\usepackage{subfigure}
\usepackage{cite}
\usepackage{booktabs}
\usepackage{amsmath,amssymb,amsfonts}
\usepackage{amsthm}
\usepackage{algorithmic}
\usepackage{graphicx}
\usepackage{textcomp}
\usepackage{xcolor}
\usepackage{makecell}
\usepackage{multirow}
\usepackage{todonotes}
\usepackage{bm}
\usepackage{tcolorbox}
\usepackage{mathtools}
\usepackage{diagbox}
\usepackage{array}
\usepackage{verbatim}
\usepackage{bm}
\usepackage{upgreek}
\usepackage{tabularx}
\usepackage{cleveref}
\usepackage{graphicx}
\newtheorem{proposition}{Proposition}
\newtheorem{definition}{Definition}

\newtheorem{lemma}{Lemma}
\theoremstyle{remark}
\allowdisplaybreaks
\newtheorem{remark}{Remark}

\def\BibTeX{{\rm B\kern-.05em{\sc i\kern-.025em b}\kern-.08em
    T\kern-.1667em\lower.7ex\hbox{E}\kern-.125emX}}
\usepackage{balance}

\makeatletter
\renewcommand{\citepunct}{,\penalty\@m\hskip.13emplus.1emminus.1em}
\renewcommand{\citedash}{\hbox{--}\penalty\@m}
\makeatother

\begin{document}

\title{Joint Optimization of Uplink and Downlink Resources under QoS Constraints of AR}

\author{Shiyong~Chen,~\IEEEmembership{Student Member,~IEEE,} and Shengqian~Han,~\IEEEmembership{Senior Member,~IEEE}%
\thanks{Shiyong Chen is with the School of Electronics and Information Engineering, Beihang University, Beijing 100191, China (email: shiyongchen@buaa.edu.cn).}
\thanks{Shengqian Han is with the School of Electronics and Information Engineering, Beihang University, Beijing 100191, China (email: sqhan@buaa.edu.cn).}
}

\maketitle

\begin{abstract}
This paper studies joint uplink (UL) and downlink (DL) resource optimization for interactive augmented reality (AR) services, where the live video captured by an AR device is uploaded to the network edge, and then the augmented video is subsequently downloaded. By modeling the AR transmission process as a tandem queuing system, we derive an upper bound for the probabilistic quality of service (QoS) requirement concerning end-to-end latency and reliability. 
The derived bound transforms the probabilistic QoS requirement into a tractable service-time condition that jointly characterizes the UL and DL service processes. Based on this condition, we formulate a weighted UL-DL transmit-power minimization problem and propose a learning-based framework to jointly optimize UL power allocation and DL beamforming. To enable gradient-based training, we further derive a differentiable upper bound for the service-time condition. Moreover, we design GNN-based policies for UL power allocation and DL beamforming, where the UL GNN exploits permutation equivariance (PE) and the DL GNN incorporates both PE and the optimal structure of wideband DL beamforming. Simulation results show that the proposed method satisfies the AR reliability requirement and reduces the weighted transmit power compared with baselines that optimize UL and DL resources separately.
\end{abstract}


\section{Introduction}
Deploying augmented reality (AR) over wireless networks is a crucial step towards realizing the Metaverse~\cite{metaverse}. AR integrates virtual objects into a live view of the real world, creating a realistic and personalized interactive environment. To achieve a seamless, immersive wireless AR experience, high data rate is required and stringent quality-of-service (QoS) requirements concerning end-to-end (E2E) latency and reliability should be~satisfied.

For AR services, the E2E latency requirement is modeled as the packet delay budget (PDB), where the delay budget defines the maximum allowable delay from the instant a live video frame is generated to the instant the corresponding augmented video frame is returned. Reliability can be modeled by the packet loss rate (PLR), which includes both the probability of packet transmission error and the probability that the E2E delay exceeds the PDB, known as the PDB violation probability~\cite{Standard}. In the mobile edge computing (MEC)-based wireless AR system, resource allocation to ensure latency and reliability requirements was investigated in \cite{Resource_allocation, QoS_Aware, Task_Offloading}. A federated learning approach was proposed in \cite{Resource_allocation} to minimize resource usage, where the PDB was treated as a hard constraint to ensure that the E2E latency of any packet does not exceed the PDB. In~\cite{QoS_Aware, Task_Offloading}, the total E2E delay was minimized under resource constraints. However, due to the fluctuation of wireless channels, using PDB as a hard constraint can result in unbounded resource utilization under poor channel conditions. Moreover, minimizing E2E delay may excessively satisfy the E2E latency requirement, leading to resource waste. By taking PDB as the latency constraint and allowing rare PDB violations, the efficiency of resource utilization can be significantly improved \cite{Probabilistic_constrained}.

\subsection{Related Works}
Given the interactive nature of AR services, which involve uplink (UL) transmission, edge computing, and downlink (DL) transmission, an AR system can be modeled as a tandem queueing system. The PDB violation probability of such a system can be analyzed using stochastic network calculus (SNC), which provides upper bounds by representing complex queueing dynamics as analytically tractable linear models in the min-plus or max-plus algebra framework~\cite{SNC_jiang,Two_side,Time_domain,Time_domain1}. 

Based on SNC, several resource-allocation methods have been developed to satisfy probabilistic delay constraints. In~\cite{QoS_Guaranteed}, a martingale-based SNC model was used to derive stochastic upper bounds on transmission delay, and a greedy resource-block (RB) reallocation algorithm was developed to improve fairness among radio slices with heterogeneous QoS requirements. 
For uplink NOMA, an SNC-based upper bound on the queueing delay violation probability was derived in~\cite{Uplink_NOMA}. Using this bound as a statistical delay QoS constraint, the sum transmit power of a NOMA user pair was minimized through a power-allocation algorithm combining one-dimensional search, gradient descent, and bisection.
In integrated sensing and communication networks~\cite{SNC_Analysis}, the transmission-delay upper bound of sensory data was derived using SNC. Based on this bound, the power allocation coefficient between sensing and communication was optimized, with the objective of minimizing the delay upper bound under sensing-quality constraints. SNC has also been applied to resource allocation in ultra-reliable and low-latency communication (URLLC) systems~\cite{Statistical_QoS,Effective_energy_efficiency,AoI_aware,AoI_Driven,AoI_Driven_Statistical,Performance_evaluation}. For example, in~\cite{Effective_energy_efficiency}, transmission reliability was characterized by both the delay violation probability and the average decoding error probability, and a power-control problem was formulated to maximize effective energy efficiency under reliability and power constraints. The transmit power were iteratively optimized using bisection search and gradient descent, respectively. Although these methods can enforce probabilistic delay constraints through SNC-based upper bounds, they cannot be directly applied to AR systems because they are mainly developed for queueing systems with a single service node.

AR systems are modeled as tandem queueing systems with multiple service nodes~\cite{Roundtrip_Interaction,uVR,Learn_to_Optimize}, making wireless resource optimization under probabilistic delay constraints more challenging. In~\cite{uVR}, an upper bound on the E2E delay violation probability of a VR system was derived by combining SNC with martingale theory, and the theoretical bound was validated by comparing the resulting communication reliability with simulation results. However, no explicit resource-allocation optimization problem was formulated. Beyond delay-bound derivation, power allocation has also been studied for tandem queueing systems under stochastic delay requirements~\cite{Learn_to_Optimize, Martingale_Based}. Nevertheless, these methods typically decompose the E2E probabilistic delay constraint into separate constraints for individual service nodes, and then allocate resources at each node independently. Such a decomposition prevents joint wireless resource allocation across nodes.

The joint optimization of wireless resources under E2E probabilistic delay constraints often leads to high dimensional and non-convex problems~\cite{Learn_to_Optimize, Martingale_Based}, motivating learning-based resource-allocation methods. To reduce the training complexity of such methods, recent studies have exploited mathematical properties of target policies, such as permutation equivariance (PE), for DNN design. For example, various graph neural networks (GNNs) have been designed to leverage different PE properties for power allocation~\cite{Optimal_Wireless_Resource_Allocation, Learning_power, GNNs_for_Scalable_Radio, Learn_to_Optimize} and beamforming~\cite{Understanding_the_Performance, Designing_Heterogeneous_GNNs, Learning_Hybrid_Precoding,A_Model_Based,A_Gradient_Driven, Gradient_Driven_Graph, Gradient_GNN, A_Model_based_DNN}. 
Another approach is to incorporate the structural properties of optimal solutions into DNN design. The optimal solution structure of narrowband multiuser downlink beamforming was characterized in~\cite{Optimal_Structure}, showing that the full beamforming matrix can be recovered from low-dimensional power-related vectors. Based on this structure, model-driven learning methods have been developed for beamforming, where the DNN learns low-dimensional power features rather than the high-dimensional beamforming matrix, and the beamforming matrix is then analytically reconstructed~\cite{Joint_User_Scheduling,Model_Driven,A_Data_and, A_Bipartite_Graph, Transfer_Learning_and}. Although this structure can reduce the training complexity of learn-based beamforming, it cannot be directly applied to AR systems with wideband channels, where the aggregate-rate and QoS constraints couple power allocation across~subchannels.

\subsection{Motivation and Contributions}
Although existing SNC-based methods can analyze the PDB violation probability of AR tandem queueing systems, they usually decompose the E2E probabilistic delay requirement into separate constraints for individual service nodes. Such a decomposition neglects the coupling between UL and DL service processes and prevents joint wireless resource allocation across nodes. To address this limitation, we derive a new PDB violation probability upper bound that characterizes the joint behavior of the UL and DL service processes through a unified service-time condition. Based on this bound, we formulate a joint UL power allocation and DL beamforming problem under probabilistic QoS constraints. To solve the resulting high-dimensional and non-convex problem, we develop an learning method that incorporates both PE property and the optimal structure of the wideband DL beamforming. 
The main contributions of this paper are summarized as~follows.\footnote{A part of this work, specifically the tandem queuing model of the AR interaction process, was reported in a conference paper~\cite{Learn_to_Optimize}. This manuscript substantially extends~\cite{Learn_to_Optimize} by formulating a new weighted UL-DL transmit-power minimization problem, deriving a new PDB violation probability upper bound for joint UL-DL resource allocation, and developing GNN-based UL power allocation and DL beamforming policies that exploit PE properties and the optimal solution structure.}

\begin{itemize}
    \item We derive an upper bound on the PDB violation probability of the AR system using SNC and Doob's inequality. The derived bound yields a unified service-time condition that jointly characterizes the UL and DL service processes, enabling joint UL-DL resource allocation under the unified service-time condition.

    \item We formulate a weighted UL-DL transmit-power minimization problem under the derived service-time condition and develop a learning framework to jointly optimize UL power allocation and DL beamforming. To enable gradient-based training, we further derive a differentiable upper bound for the service-time condition.

    \item We design GNNs for UL power allocation and DL beamforming policies. The UL GNN exploits user PE, while the DL GNN incorporates both PE and the optimal structure of the wideband DL beamforming, reducing the output dimension from the full beamforming matrix to low-dimensional power vectors. Simulation results verify the tightness of the derived PDB bound and show that the proposed method satisfies the AR reliability requirement and reduces the weighted transmit power compared with baselines that optimize UL and DL resources separately.
\end{itemize}

\textit{Notations:} $(\cdot)^{\mathsf T}$, $(\cdot)^{\mathsf H}$, and $(\cdot)^{*}$  denote the transpose, Hermitian transpose and conjugate, respectively. $|\cdot|$ and $|\cdot|_1$ denote Frobenius and $\ell_1$ norm. $\mathbf{I}_z$ is the identity matrix of size $z\times z$, $\succeq$ indicates element-wise inequality, \(\overline{\otimes}\) denotes max-plus
convolution, and \(\lceil x\rceil\) denotes the smallest integer no less
than \(x\).

\section{QoS Requirement and System Model}
Consider a MEC-assisted wireless AR system, where a MEC-enabled base station (BS) equipped with $N_{t}$ antennas serves $K$ single-antenna AR users (AUs). 
Due to the limited computation and power resources of the AR devices, the computation tasks, such as object detection and rendering, are offloaded to the MEC. This requires AUs to upload their live videos to the BS via UL transmission. The MEC then detects target objects in the received videos, generates virtual objects, and superimposes them onto the detected objects. The augmented videos are compressed and delivered back to the AUs via DL transmission. Since medium-quality UL video is sufficient for the detection of target objects, the UL video streams are often downscaled to reduce transmission requirements compared to the higher-quality DL streams \cite{Standard}.

\subsection{QoS Requirement of AR System}
AR services have strong requirements for low E2E latency, high reliability, and high data rates. The transmission between the $\mathrm{AU}_k$ and the BS can be modeled as a tandem queuing system in the time domain, as shown in Fig.~\ref{Tandem queueing model}, where the $\mathrm{AU}_k$ and the BS are service nodes. In the UL, the live video is segmented into frames, and the bits in a frame are represented as a packet. The average inter-arrival time between packets is the inverse of the frame rate $f$. The instantaneous inter-arrival time $\tau_k\left(n\right)$ between the $n$-th packet and the $(n+1)$-th packet is random due to jitter, which follows a truncated Gaussian distribution with mean $\mu$, variance $\sigma^2$ and lies within the interval $\left(b_1, b_2\right)$ according to 3GPP~specifications~\cite{Standard}.

\begin{figure}[htbp]
\centerline{\includegraphics[width=0.45\textwidth]{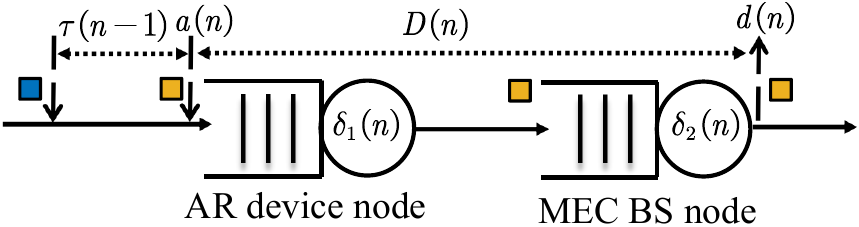}}
\caption{Tandem queueing model of the AR system.} 
\vspace{-0.2cm}
\label{Tandem queueing model}
\end{figure}

Given the randomness in the packet arrival process and the fluctuation of UL and DL wireless channels, packets may accumulate in queues at both the AR devices and the BS. For $\mathrm{AU}_k$, let $a_k\left(n\right)$ denote the arrival time of the $n$-th packet, and $d_k\left(n\right)$ denote the time when the augmented video frame corresponding to the $n$-th packet is transmitted to the $\mathrm{AU}_k$. Then, the E2E latency of the $n$-th packet can be expressed as
\vspace{-0.1cm}
\begin{equation}
D_k\left( n \right) =d_k\left( n \right) -a_k\left( n \right) ,
\label{Definition of delay}
\end{equation}	
where $D_k\left( n\right)$ includes UL transmission delay, UL queuing delay, DL transmission delay, and DL queuing delay. The computing delay at the BS is negligible compared to the PDB, due to the deployment of high-performance hardware and rendering algorithms, and is thus omitted here~\cite{Prediction_Communication_and_Computing}.

For AR services, receiving a packet late is nearly as detrimental as losing it entirely \cite{Standard}. Consequently, reliability is often characterized by the packet loss rate (PLR), which includes both the probability of packet transmission errors and the probability that a packet's E2E latency exceeds the PDB. Owing to the strong error‑correction capability of channel coding and the employment of closed‑loop link‑adaptation techniques, packet transmission errors are effectively controlled. As a result, the PLR is primarily determined by the PDB violation probability~\cite{uVR}. Thus, the QoS related to latency, reliability, and data rates is defined as
\begin{align}
  P\left( D_k\left( n \right) \geq D^k_{\max} \right) \leq \varepsilon^k_{\max}, \label{eq:ARqos}
\end{align}
where $D^k_{\max}$ and $\varepsilon^k_{\max}$ denotes the PDB and the target PLR for $\mathrm{AU}_k$, respectively.

\subsection{Transmission Model}
To ensure high data rate requirements of the AR system, we consider a wideband time division multiplexing transmission system, with \(M_{\mathrm{u}}\) UL subchannels and \(M_{\mathrm{d}}\) DL subchannels. The bandwidth of each subchannel is $W_0$. In the time domain, transmission time is divided into coherence blocks of duration $T_{\mathrm{c}}$. Each coherence block is further divided into an UL interval $T_{\mathrm{u}}$ and an DL interval $T_{\mathrm{d}}$, with $T_{\mathrm{c}}=T_{\mathrm{u}}+T_{\mathrm{d}}$. The large-scale channel gains are assumed to remain constant across coherence blocks while the small-scale channel gains are independent and identically distributed (i.i.d.) across different coherence blocks and remain constant within each block. 

In the UL transmission, each AU transmits its live video stream to the BS. The data rate of $\mathrm{AU}_k$ at the $t$-th coherence block is expressed as 
\begin{equation} \label{UL capacity}
 R_{k,t}^{\mathrm{u}}=\sum\limits_{m=1}^{M_{\mathrm{u}}}{W_0\log _2\big( 1+\gamma_{m,k,t}^{\mathrm{u}}\big)},
\end{equation}
where $\gamma_{m,k,t}^{\mathrm{u}}$ denotes the signal-to-interference-plus-noise ratio (SINR)
of $\mathrm{AU}_k$ on subchannel $m$ during the \( t \)-th block. Using the minimum mean square error (MMSE) receiver at the BS, the UL SINR is given by~\cite{Learning_uplink_Power}
\begin{equation}
\begin{split}
 \!\!\!\!\gamma^{\mathrm{u}}_{m,k,t}\!=\!\!\frac{p_{m,k}\alpha _k\left| \mathbf{w}_{m,k,t}^{\mathsf{H}}\mathbf{h}_{m,k,t} \right|^2}{\!\!  \!\!\sum\limits_{j=1,j\ne k}^K\!\!\!\!\!{p_{m,j}\alpha_j\!\left| \mathbf{w}_{m,k,t}^{\mathsf{H}}\mathbf{h}_{m,j,t}^{} \right|^2}\!\!\!\!+\!\!\left| \mathbf{w}_{m,k,t}^{\mathsf{H}}\mathbf{v}_{m,k,t}^{} \right|\!\sigma_{\mathrm{u}}^2}, \label{UL_SINR}
\end{split}
\end{equation}
where $\alpha_k$ denotes the large-scale channel gain of $\mathrm{AU}_k$, $\mathbf{h}_{m,k,t} \in \mathbb{C}^{N_t \times 1}$ and $\mathbf{w}_{m,k,t} \in \mathbb{C}^{N_t \times 1}$ are the instantaneous channel and MMSE beamforming vectors between $\mathrm{AU}_k$ and the BS over subchannel $m$ at the $t$-th block, respectively, $p_{m,k}$ is the transmit power allocated by $\mathrm{AU}_k$ to subchannel $m$, and $\sigma_{\mathrm{u}}^2$ is the noise power of each UL subchannel.

Similarly, the BS transmits the augmented video streams to AUs in the DL transmission. The DL data rate of $\mathrm{AU}_k$ at the $t$-th coherence block is expressed as 
\begin{equation} \label{DL capacity}
 R_{k,t}^{\mathrm{d}}=\sum\limits_{m=1}^{M_{\mathrm{d}}}{W_0 \log _2\big( 1+\gamma _{m,k,t}^{\mathrm{d}} \big)},
\end{equation}
where $\gamma_{m,k,t}^{\mathrm{d}}$ denotes the DL SINR
of $\mathrm{AU}_k$ on subchannel $m$ during the \( t \)-th block. The DL SINR is given by
\begin{equation} 
 \gamma _{m,k,t}^{\mathrm{d}}=\frac{\alpha _k\left| \mathbf{h}_{m,k,t}^{\mathsf{H}}\mathbf{v}_{m,k,t}^{} \right|^2}{\sum\nolimits_{j=1,j\ne k}^K{\alpha _k\left| \mathbf{h}_{m,k,t}^{\mathsf{H}}\mathbf{v}_{m,j,t}^{\mathsf{H}} \right|^2}+\sigma^2_{\mathrm{d}}}, \label{DL_SINR}
\end{equation}
where $\mathbf{v}_{m,k,t}\in \mathbb{C} ^{N_t \times 1}$ denotes the DL beamforming vector for $\mathrm{AU}_k$ and $\sigma_{\mathrm{d}}^2$ is the noise power of each DL subchannel.

Due to the large amount of data in each packet, transmitting a packet typically requires multiple coherence blocks. The total transmission time, referred to as service time, for the $n$-th packet of $\mathrm{AU}_k$ in the UL and DL is denoted as $\delta_k^{\mathrm{u}}\left( n \right)$ and $\delta_k^{\mathrm{d}}\left( n \right)$, respectively, which can be expressed as
\begin{subequations} \label{Service time}
    \begin{align}
    \!\!\!\delta_k^{\mathrm{u}}\left( n \right)& = \min\!\left\{\! m\!\in\!\mathbb{N}_+\!\!\left| \sum\nolimits_{t=t_{k,n}^{\mathrm{u}}}^{t_{k,n}^{\mathrm{u}}+m-1}\!{R_{k,t}^{\mathrm{u}}\!\cdot\! T_{\mathrm{u}}}\ge L_{\mathrm{u}} \right. \right\}\! \cdot\! T_{\mathrm{c}},
    \label{UL service time}\\
    \!\!\!\delta_k^{\mathrm{d}}\left( n \right)& = \min\! \left\{\! m\!\in\!\mathbb{N}_+\!\!\left| \sum\nolimits_{t=t_{k,n}^{\mathrm{d}}}^{t_{k,n}^{\mathrm{d}}+m-1}\!{R_{k,t}^{\mathrm{d}}\!\cdot\! T_{\mathrm{d}}}\ge L_{\mathrm{d}} \right. \right\}\! \cdot\! T_{\mathrm{c}},\label{DL service time}
    \end{align}
\end{subequations}
where $L_{\mathrm{u}}$ and $L_{\mathrm{d}}$ represent the sizes of the packets transmitted in the UL and DL, and $t_{k,n}^{\mathrm{u}}$ and $t_{k,n}^{\mathrm{d}}$ denote the indices of the blocks when the transmission of the $n$-th packet starts in the UL and DL.

\section{Joint UL-DL Optimization Problem}
To reduce bidirectional transmit power, we jointly optimize the UL power allocation and DL beamforming under the QoS constraints~\eqref{eq:ARqos}. The UL power allocation is determined by the large-scale channel gains, whereas the DL beamforming adapts to both the large-scale and small-scale channel gains in each coherence block. The optimization problem is formulated as
\begin{subequations} \label{P0:Power0}
    \begin{align}
\mathrm{P}0:\min_{\mathbf{P}, \mathbf{V}_{t}} \quad
& (1-\beta) |\mathrm{vec}(\mathbf{P})|_1
+\beta E_{t}\left[|\mathbf{V}_{t}|^2\right]  \label{P0:object0}\\
\mathrm{s.t.}\quad
&\mathbf{P}\succeq 0, \label{P1:UL_resction0}\\
&\sum_{m=1}^{M_{\mathrm{u}}}p_{m,k} \leq P_{\max}^{\mathrm{u}}, \quad k=1,\ldots,K,\label{P1:UL_resction1}\\
&|\mathbf{V}_t|_{\mathrm{F}}^2 \leq P_{\max}^{\mathrm{d}}, \quad \forall t\in \mathbb{N}^+,\label{P1:DL_resction}\\
&\eqref{UL capacity},\eqref{UL_SINR},
\eqref{DL capacity},\eqref{DL_SINR},\eqref{Service time},\eqref{Definition of delay},\eqref{eq:ARqos}, \nonumber
    \end{align}
\end{subequations}
where $\mathbf{P}=[p_{1,1},\ldots,p_{M_{\mathrm{u}},K}]\in\mathbb{R}^{M_{\mathrm{u}}\times K}$ denotes the UL power allocation, $\mathbf{V}_t=[\mathbf{v}_{1,1,t},\ldots,\mathbf{v}_{M_{\mathrm{d}},K,t}]\in\mathbb{C}^{M_{\mathrm{d}}\times N_t\times K}$ is the DL beamforming matrix, $P_{\max}^{\mathrm{u}}$ and $P_{\max}^{\mathrm{d}}$ are the maximum power budgets of each AU and the BS, respectively, and $\beta\in\left(0,1\right)$ denotes the weighting coefficient. In $\mathrm{P}0$, the weighted sum of UL and DL power is averaged over all time steps, referred to as the weighted UL--DL power consumption for short in the sequel.

Directly solving $\mathrm{P}0$ is challenging because the QoS requirement in~\eqref{eq:ARqos} imposes a probabilistic constraint on the E2E delay induced by the tandem queueing dynamics of the arrival and service processes. To address this constraint, the next section derives an upper bound on the PDB violation probability and identifies a sufficient service-time condition that guarantees the QoS requirement.

\subsection{PDB Violation Probability for the AR System}
The AR system is modeled as a tandem queueing system, as shown in Fig.~\ref{Tandem queueing model}. 
Applying the max-plus queueing principle~\cite{Two_side}, the E2E delay of this queueing system, as defined in~\eqref{Definition of delay}, can be rewritten as
\begin{equation}
\begin{aligned}
D(n)
&= a \overline{\otimes} \Delta_{\mathrm{u}} \overline{\otimes} \Delta_{\mathrm{d}}(n)-a(n) \\
&= \max_{0\leq m\leq n}
\left[
\Delta_{\mathrm{u}} \overline{\otimes} \Delta_{\mathrm{d}}(m,n)-\Gamma(m,n)
\right],
\end{aligned}
\label{Queueing principle}
\end{equation}
where the AU index \(k\) is omitted for notational simplicity, $\overline{\otimes}$ denotes max-plus convolution, $\Delta_{\mathrm{u}}(m,n)=\sum_{i=m}^n\delta_{\mathrm{u}}(i)$ and $\Delta_{\mathrm{d}}(m,n)=\sum_{i=m}^n\delta_{\mathrm{d}}(i)$ are the cumulative service processes of the UL and DL nodes, respectively, and $\Gamma(m,n)=\sum_{i=m}^{n-1}\tau(i)$ is the cumulative inter-arrival time. With~\eqref{Queueing principle}, the PDB violation probability is expressed~as
\begin{equation} \label{eq:pdb_prob}
\begin{split}
&P\{D(n)>D_{\max}\} \\
&=
P\Big\{
\max_{0\leq m\leq n}
e^{\theta[
\Delta_1 \overline{\otimes} \Delta_2(m,n)
-\Gamma(m,n)]}
>e^{\theta D_{\max}}
\Big\}  \\
&=P\Big\{ \max_{0\le l\le n} V\left( l \right) >e^{\theta D_{\max}} \Big\},
\end{split}
\end{equation}
where $V(l)=
e^{\theta[
\Delta_1 \overline{\otimes} \Delta_2(n-l,n)
-\Gamma(n-l,n)]}$, and $\theta>0$.

As in~\cite{Learn_to_Optimize}, an upper bound for probabilities of the form in~\eqref{eq:pdb_prob} can be obtained using Doob's inequality in Lemma~\ref{Martingale Inequality}. However, this
requires that \(\{V(l),0\leq l\leq n\}\) be a supermartingale, as defined in Definition~\ref{Supermartingale Process}. 
\begin{lemma}[Doob's Inequality]
\label{Martingale Inequality}
Let $\{X(n),n\geq0\}$ be a supermartingale. For any $x>0$,
\begin{equation}
P\Big\{\max_{0\leq m\leq n} X(m) \geq x \Big\}
\leq
\frac{E[X(0)]}{x}.
\end{equation}
\end{lemma}
\begin{IEEEproof}
See Lemma 5.18 in~\cite{Essentials_of_Stochastic_Processes}.
\end{IEEEproof}

\begin{definition}[Supermartingale Process]
\label{Supermartingale Process}
A stochastic process $\{X(n),\mathcal{F}(n),n\geq0\}$ is a supermartingale if $X(n)$ is $\mathcal{F}(n)$-measurable, $E[|X(n)|]<\infty$, and
\begin{equation}
E[X(n+1)\mid \mathcal{F}(n)]\leq X(n).
\end{equation}
\end{definition}

To derive an upper bound for~\eqref{eq:pdb_prob}, we first establish a sufficient condition under which  \(\{V(l),0\leq l\leq n\}\) froms a supermartingale.  
Let \(\mathcal{F}(l)\) denote the filtration, i.e., the information available up to step \(l\). With \(\mathcal{F}(l)\), the expectation of $V(l+1)$ is denoted as
\begin{equation}
\begin{split}\raisetag{0.5cm}
&E\left[V(l+1)\!\mid \!\mathcal{F}(l)\right]\!=\!
E\!\left[\!
e^{\theta[
\Delta_{\mathrm{u}}\overline{\otimes}\Delta_{\mathrm{d}}(n-l-1,n)
-\Gamma(n-l-1,n)]}
\!\!\mid\!\! \mathcal{F}(l)
\!\right] \\
&\underset{(a)}{\leq}
V(l)
E\left[
e^{\theta\{\max[\delta_{\mathrm{u}}(n-l-1),\delta_{\mathrm{d}}(n-l-1)]
-\tau(n-l-1)\}}
\right] \\
&=
V(l)
E\left[
e^{\theta\{\max[\delta_{\mathrm{u}}(n),\delta_{\mathrm{d}}(n)]-\tau(n)\}}
\right] \\
&\underset{(b)}{\leq}
V(l).
\end{split}
\label{Supermartingale of tandem}
\end{equation}
where $(a)$ follows because adding one packet to the tandem service process increases the cumulative service time by no more than the larger service time of the two nodes, and inequality $(b)$ holds if
\begin{equation}
E\Big[
e^{\theta\{\max[\delta_{\mathrm{u}}(n),\delta_{\mathrm{d}}(n)]-\tau(n)\}}
\Big]\leq 1.
\label{eq:condition}
\end{equation}
Therefore, according to Definition~\ref{Supermartingale Process}, $\{V(l),0\leq l\leq n\}$ forms a supermartingale under \eqref{eq:condition}.

Then applying Lemma~\ref{Martingale Inequality} to~\eqref{eq:pdb_prob}, the violation probability is upper bounded by
\begin{equation}
\begin{aligned}
&P\{D(n)>D_{\max}\}
\leq
E[V(0)]e^{-\theta D_{\max}}
 \\
&=E\left[e^{\theta\delta_{\mathrm{u}}(n)}\right]
E\left[e^{\theta\delta_{\mathrm{d}}(n)}\right]
e^{-\theta D_{\max}}
 \\
&\leq
\frac{1}{E\left[e^{-\theta\tau(n)}\right]^2}e^{-\theta D_{\max}}
\triangleq f\left(\theta, D_{\max} \right),
\end{aligned}
\label{eq:Bound of tandem}
\end{equation}
where the last inequality follows from~\eqref{eq:condition}.

With~\eqref{eq:Bound of tandem}, the QoS requirement of the AR system in~\eqref{eq:ARqos} is relaxed to the constraint on the upper bound of PDB violation probability, denoted as
\begin{align}
  f\left(\theta_k,D^k_{\max}\right) \leq \varepsilon^k_{\max},
  \label{eq:qos upper bound}
\end{align}
where the QoS exponent \(\theta_k>0\), and the UL and DL service process satisfies
\begin{equation}
 E\Big[
e^{\theta_k\{\max[\delta^k_{\mathrm{u}}(n),\delta^k_{\mathrm{d}}(n)]
-\tau_k(n)\}}
\Big]\leq 1.
\label{eq:condition for AUk}
\end{equation}

\begin{remark}
The bound in~\eqref{eq:Bound of tandem} is valid when the service process satisfies~\eqref{eq:condition}. This service-time condition constraints the joint behavior of the UL and DL service processes through \(\max\{\delta_{\mathrm{u}}(n),\delta_{\mathrm{d}}(n)\}\), instead of imposing separate constraints to the two service nodes as in~\cite{Learn_to_Optimize,A_combined,Dealing_wit}. This joint constraint enables the subsequent joint optimization of UL power allocation and DL beamforming.
\end{remark}

\subsection{Problem Reformulation}
Although the relaxation in~\eqref{eq:qos upper bound} provides a tractable way to enforce the probabilistic QoS constraint~\eqref{eq:ARqos} in \(\mathrm{P}0\), its tightness is influenced by the QoS exponent $\theta_k$. As the tightness affect the power allocation and beamforming, we first determine the optimal value of $\theta_k$. Noting that increased PDB violation probability is associated with lower transmit power, the optimal $\theta_k^\star$ is achieved when \eqref{eq:qos upper bound} holds with equality, i.e.,
\begin{align}
  f\left(\theta_k, D^k_{\max} \right) = \varepsilon^k_{\max}. \label{eq:optimal theta}
\end{align}
Since \(D^k_{\max}\), \(\varepsilon^k_{\max}\), and the distribution of \(\tau_k(n)\) are given, the optimal value $\theta^\star$ can be readily obtained by a
one-dimensional search, such as bisection.

By setting \(\theta_k=\theta_k^\star\), the probabilistic QoS constraint in~\eqref{eq:ARqos} is replaced by the deterministic service-time condition in~\eqref{eq:condition for AUk}. The original problem \(\mathrm{P}0\) is then reformulated as
\begin{subequations}
\label{P1:Power0}
    \begin{align}
\mathrm{P}1:\min_{\mathbf{P}, \mathbf{V}_{t}} \quad
& (1-\beta) |\mathrm{vec}(\mathbf{P})|_1
+\beta E_{t}\left[|\mathbf{V}_{t}|^2\right]  \label{P1:object0}\\
\mathrm{s.t.}\quad
&\mathbf{P}\succeq 0, \label{P1:UL_resction0}\\
&\sum_{m=1}^{M_{\mathrm{u}}}p_{m,k} \leq P_{\max}^{\mathrm{u}}, \quad k=1,\ldots,K,\label{P1:UL_resction1}\\
&|\mathbf{V}_t|^2 \leq P_{\max}^{\mathrm{d}}, \quad \forall t\in \mathbb{N}^+,\label{P1:DL_resction}\\
&\eqref{UL capacity},\eqref{UL_SINR},
\eqref{DL capacity},\eqref{DL_SINR},\eqref{Service time},
\eqref{eq:condition for AUk},\eqref{eq:optimal theta}.\nonumber
    \end{align}
\end{subequations}
where the QoS exponent $\theta_k$ is determined by~\eqref{eq:optimal theta}. Problem \(\mathrm{P}1\) minimizes the same weighted UL--DL power consumption as
\(\mathrm{P}0\), while replacing the probabilistic QoS constraint~\eqref{eq:ARqos} with the deterministic service-time condition~\eqref{eq:condition for AUk}.

\section{Learning for Joint UL--DL Optimization}
While the probabilistic QoS constraint has been replaced with the deterministic service-time condition in~\eqref{eq:condition for AUk} in \(\mathrm{P}1\), directly solving this problem remains nontrivial. The difficulty stems from the service times $\delta^k_{\mathrm{u}}(n)$ and $\delta^k_{\mathrm{d}}(n)$ in~\eqref{eq:condition for AUk}, which are discrete quantities obtained through enumeration in~\eqref{Service time}. Consequently, \(E\!\big[
e^{\theta_k\max\{\delta^k_{\mathrm{u}}(n),\delta^k_{\mathrm{d}}(n)\}}\big]\) in the deterministic condition is not differentiable with respect to \(\mathbf{P}\) and \(\mathbf{V}_t\), which prevents direct gradient-based optimization.

To address this challenge, we employ an unsupervised primal-dual learning framework to jointly optimize \(\mathbf{P}\) and \(\mathbf{V}_t\), where the the loss function is the Lagrangian function of \(\mathrm{P}1\). To satisfy the QoS constraint, the service-time condition~\eqref{eq:condition for AUk} is incorporated as a penalty term in the loss function. 
To enable gradient-based training, we further derive a differentiable upper bound on \(E\!\big[
e^{\theta_k\max\{\delta^k_{\mathrm{u}}(n),\delta^k_{\mathrm{d}}(n)\}}\big]\). 

\subsection{Unsupervised Learning with the QoS Constraint}
We first reformulate \(\mathrm{P}1\) into a primal-dual form, which allows the primal and dual variables to be jointly updated within an unsupervised primal-dual learning framework~\cite{Unsupervised}. Denote the UL power allocation and DL beamforming policy~as
\begin{equation}\label{eq:UL and DL policy}
\mathbf{P}
=
\upphi_{\mathrm{u}}\left(\boldsymbol{\alpha},
\boldsymbol{\theta}^{\star}\right),
\qquad
\mathbf{V}_t
=
\upphi_{\mathrm{d}}\left(\mathbf{H}_t,
\boldsymbol{\theta}^{\star}\right),
\end{equation}
where
\(\boldsymbol{\theta}^{\star}
=[\theta^{\star}_1,\ldots,\theta^{\star}_K]^{\mathsf T}\in\mathbb{R}^{K\times1}\)
denotes the optimal QoS exponent vector obtained from~\eqref{eq:optimal theta},
\(\boldsymbol{\alpha}=[\alpha_1,\ldots,\alpha_K]^{\mathsf T}\in\mathbb{R}^{K\times1}\) is the vector of large-scale channel gains, and \(\mathbf{H}_t=[\sqrt{\alpha_1}\mathbf{h}_{1,1,t},\ldots,\sqrt{\alpha_K}\mathbf{h}_{M_{\mathrm d},K,t}]
\in\mathbb{C}^{M_{\mathrm d}\times N_t\times K}\) is the composite channel matrix. The primal-dual form of \(\mathrm{P}1\) is formulated as
\begin{subequations} \label{P3:primal-dual problem}
    \begin{align}
    \underset{\boldsymbol{\lambda}}{\max}\,&\underset{\upphi_{\mathrm{u}}\left( \boldsymbol{\alpha}, \boldsymbol{\theta}^{\star} \right), \upphi_{\mathrm{d}}\left( \mathbf{H}_{t}, \boldsymbol{\theta}^{\star} \right)}{\!\!\!\min}\!\!\!\mathcal{L} \left( \upphi_{\mathrm{u}}\left(\boldsymbol{\alpha }, \boldsymbol{\theta}^{\star}\right) ,\upphi_{\mathrm{d}}\left(\mathbf{H}_{t},\boldsymbol{\theta}^{\star}\right),\boldsymbol{\lambda} \right)   \label{P3:object}\\
     \mathrm{s}.\mathrm{t}. \  &\upphi_{\mathrm{u}}\left(\boldsymbol{\alpha}, \boldsymbol{\theta}^{\star}\right)  \succeq 0, \label{P3:power resction1}\\
     &\sum_{m=1}^{M_{\mathrm{u}}}{p_{m,k}} \leq P_{\max}^{\mathrm{u}}, \forall k \in\{1,2,\cdots, K\}, \label{P3:power resction2}\\
     &| \upphi_{\mathrm{d}}\left(\mathbf{H}_{t}, \boldsymbol{\theta}^{\star}\right) |^2 \leq P_{\max}^{\mathrm{d}}, \forall t\in \mathbb{Z} ^+,  \label{P3:precoder resction1}\\
     &\boldsymbol{\lambda} \succeq 0,  \label{P3:lamada resction}
    \end{align}
\end{subequations}
where \(\boldsymbol{\lambda}=[\lambda_1,\ldots,\lambda_K]^{\mathsf T}\in
\mathbb{R}^{K\times1}\) is the Lagrange multiplier vector. The Lagrangian in \eqref{P3:object} is given by
\begin{equation}
\begin{split} \label{Lagrangian}
&\mathcal{L} \left( \upphi_{\mathrm{u}}\left(\boldsymbol{\alpha }, \boldsymbol{\theta}^{\star}\right) ,\upphi_{\mathrm{d}}\left(\mathbf{H}_{t},\boldsymbol{\theta}^{\star}\right),\boldsymbol{\lambda} \right)\\
&=(1-\beta)|\mathrm{vec}(\upphi_{\mathrm{u}}\left(\boldsymbol{\alpha}, \boldsymbol{\theta}^{\star}\right))|_1 + \beta E_{t}\left[ | \upphi_{\mathrm{d}}\left( \mathbf{H}_{t},\boldsymbol{\theta}^{\star} \right) |^2\right] \\ 
&+ \sum\nolimits_{k=1}^K{\lambda _k\Big(E\Big[
e^{\theta_k\{\max[\delta^k_{\mathrm{u}}(n),\delta^k_{\mathrm{d}}(n)]-\tau_k(n)\}}
\Big]-1\Big)}.  
\end{split}  
\raisetag{0.5cm}
\end{equation}

Since \(\upphi_{\mathrm{u}}(\cdot)\) and \(\upphi_{\mathrm{d}}(\cdot)\) are functional
optimization variables, directly optimizing them with conventional numerical
methods is difficult~\cite{Unsupervised, Functional_opti}. Thus, we parameterize them by DNNs~as
\begin{equation}\label{eq:UL and DL policy networks}
    \Phi_{\mathrm{u}}
\left(\boldsymbol{\alpha},\boldsymbol{\theta}^{\star};
\boldsymbol{\omega}_{\mathbf{P}}\right),
\qquad
\Phi_{\mathrm{d}}
\left(\mathbf{H}_{t},\boldsymbol{\theta}^{\star};
\boldsymbol{\omega}_{\mathbf{V}}\right),
\end{equation}
which are referred to as the UL and DL policy networks, respectively, with
\(\boldsymbol{\omega}_{\mathbf{P}}\) and \(\boldsymbol{\omega}_{\mathbf{V}}\)
denoting trainable parameters. Substituting~\eqref{eq:UL and DL policy networks} into ~\eqref{P3:primal-dual problem} results a parameterized primal-dual~problem
\begin{subequations} \label{P4:primal-dual problem}
    \begin{align}
    \underset{\boldsymbol{\lambda}}{\max}&\,\,\underset{\boldsymbol{\omega }_{\mathbf{P}}, \boldsymbol{\omega }_{\mathbf{V}}}{\min}\,\,\mathcal{L} \left( \Phi_{\mathrm{u}}\left( \boldsymbol{\alpha},\boldsymbol{\theta}^{\star};\boldsymbol{\omega }_{\mathbf{P}} \right),\Phi_{\mathrm{d}}\left(\mathbf{H}_{t},\boldsymbol{\theta}^{\star};\boldsymbol{\omega }_{\mathbf{V}} \right),\boldsymbol{\lambda} \right) \,\, \label{P4:object}\\
     \mathrm{s}.\mathrm{t}. \  &\Phi_{\mathrm{u}}\left( \boldsymbol{\alpha },\boldsymbol{\theta}^{\star};\boldsymbol{\omega}_{\mathbf{P}} \right)  \succeq 0, \label{P4:power resction1}\\
     &\sum_{m=1}^{M_{\mathrm{u}}}p_{m,k} \leq P_{\max}^{\mathrm{u}}, \forall k \in\{1,2,\cdots, K\}, \label{P4:power resction2}\\
     &| \Phi_{\mathrm{d}}\left(\mathbf{H}_{t},\boldsymbol{\theta}^{\star}; \boldsymbol{\omega_{\mathbf{V}}} \right) |_{\mathrm{F}}^2 \leq P_{\max}^{\mathrm{d}}, \forall t\in \mathbb{Z} ^+,  \label{P4:precoder resction1}\\
     &\boldsymbol{\lambda} \succeq 0,  \label{P4:lamada resction}
    \end{align}
\end{subequations}
where the Lagrangian function~\eqref{P4:object} is given by
\begin{equation}
\begin{split} \label{eq:Lagrangian DNNs}
&\mathcal{L} \left( \Phi_{\mathrm{u}}\left(\boldsymbol{\alpha}, \boldsymbol{\theta}^{\star};\boldsymbol{\omega }_{\mathbf{P}}\right) ,\Phi_{\mathrm{d}}\left(\mathbf{H}_{t},\boldsymbol{\theta}^{\star};\boldsymbol{\omega }_{\mathbf{V}}\right),\boldsymbol{\lambda} \right))\\
&=\left(1-\beta\right)|\Phi_{\mathrm{u}}\left(\boldsymbol{\alpha}, \boldsymbol{\theta}^{\star};\boldsymbol{\omega }_{\mathbf{P}}\right)|_1 + \beta E_{t}\left[ | \Phi_{\mathrm{d}}\left( \mathbf{H}_{t},\boldsymbol{\theta}^{\star};\boldsymbol{\omega }_{\mathbf{V}} \right) |^2\right] \\ 
&+ \sum\nolimits_{k=1}^K{\lambda _k\Big(E\Big[
e^{\theta_k\{\max[\delta^k_{\mathrm{u}}(n),\delta^k_{\mathrm{d}}(n)]-\tau_k(n)\}}
\Big]-1\Big)}.
\end{split} 
\raisetag{0.5cm}
\end{equation}

To solve problem~\eqref{P4:primal-dual problem}, the primal variables \(\boldsymbol{\omega}_{\mathbf{P}}\) and \(\boldsymbol{\omega}_{\mathbf{V}}\) and the dual variable \(\boldsymbol{\lambda}\) are iteratively updated using stochastic gradient descent and ascent, respectively. For simplicity, we define \( \bar{\mathbf{P}} \triangleq \Phi_{\mathrm{u}}\left( \boldsymbol{\alpha}, \boldsymbol{\theta}^{\star}; \boldsymbol{\omega}_{\mathbf{P}} \right) \) and \( \bar{\mathbf{V}} \triangleq \Phi_{\mathrm{d}}\left( \mathbf{H}_t, \boldsymbol{\theta}^{\star}; \boldsymbol{\omega}_{\mathbf{V}} \right) \). The updates for \( \boldsymbol{\omega}_\mathbf{P} \), \( \boldsymbol{\omega}_\mathbf{V} \) follow the descent direction of the sample-averaged gradient, derived as
\begin{subequations} \label{eq:Gradient of DNNs}
    \begin{align}
    &\boldsymbol{\omega }_{\mathbf{P}}^{(i+1)}\gets \boldsymbol{\omega }_{\mathbf{P}}^{(i)}-\frac{\xi_{\bar{\mathbf{P}}}}{|\mathcal{B} |}\sum_{\boldsymbol{\alpha },\mathbf{H}_t\in \mathcal{B}}{\nabla _{\boldsymbol{\omega }_{\mathbf{P}}}\bar{\mathbf{P}}\cdot \left( \nabla _{\bar{\mathbf{P}}}\mathcal{L} \right) ^T} \\
    &\boldsymbol{\omega }_{\mathbf{V}}^{(i+1)}\gets \boldsymbol{\omega }_{\mathbf{V}}^{(i)}-\frac{\xi_{\bar{\mathbf{V}}}}{|\mathcal{B} |}\sum_{\boldsymbol{\alpha },\mathbf{H}_t\in \mathcal{B}}{\nabla _{\boldsymbol{\omega }_{\mathbf{V}}}\bar{\mathbf{V}}\cdot \left( \nabla _{\bar{\mathbf{V}}}\mathcal{L} \right) ^T}
    \end{align}
\end{subequations}
where $\left(\cdot\right)^T$ is the transpose operation, the superscript $i$ refers to the $i$-th iteration, $\mathcal{B}$ represents a batch of samples for $\boldsymbol{\alpha }$ and $\mathbf{H}_t$, $|\mathcal{B} |$ is the batch size, $\xi _{\bar{\mathbf{P}}}$ and $\xi_{\bar{\mathbf{V}}}$ represent the learning rates of the respective DNNs. The gradients $\nabla _{\boldsymbol{\omega }_{\mathbf{P}}}\bar{\mathbf{P}}$ and $\nabla _{\boldsymbol{\omega }_{\mathbf{V}}}\bar{\mathbf{V}}$ are the transposed Jacobian matrices, which can be computed via backpropagation. To compute the gradient matrices $\nabla _{\bar{\mathbf{P}}}\mathcal{L}
$ and $\nabla_{\bar{\mathbf{V}}}\mathcal{L}$ from the Lagrangian function in \eqref{eq:Lagrangian DNNs}, the following gradients must be obtained
\begin{subequations} \label{eq:gradients of expection}
    \begin{align}    
    &\nabla_{\bar{\mathbf{P}}}E_n\Big[ e^{\theta_k\cdot \max \{\delta^k _{\mathrm{u}}\left( n \right) ,\delta^k _{\mathrm{d}}\left( n \right)\}} \Big],\\   
    &\nabla_{\bar{\mathbf{V}}}E_n\Big[ e^{\theta_k\cdot \max \{\delta^k _{\mathrm{u}}\left( n \right) ,\delta^k _{\mathrm{d}}\left( n \right)\}} \Big].
    \end{align}
\end{subequations}

The dual variable \( \lambda _{k} \) is updated in the direction of gradient ascent as follows
\begin{equation}
\begin{aligned} \label{eq:Gradients of lambda}
\lambda _{k}^{(i+1)}\!\!\gets\!\!\lambda _{k}^{(i)}\!+\!\xi _{\lambda}\Big(E\big[
e^{\theta_k\{\max[\delta^k_{\mathrm{u}}(n),\delta^k_{\mathrm{d}}(n)]-\tau_k(n)\}}
\big]-1\Big),
\end{aligned} 
\end{equation}
where \( \xi_{\lambda} \) denotes the step size for \( \lambda_k \) in the direction of gradient ascent.

While the primal and dual variables are updated by stochastic gradient method, the updates in~\eqref{eq:Gradient of DNNs} and~\eqref{eq:Gradients of lambda} require evaluating
\(E_n\left[e^{\theta_k\max\{\delta^k_{\mathrm{u}}(n),\delta^k_{\mathrm{d}}(n)\}}\right]\)
and its gradients with respect to \(\mathbf{P}\) and \(\mathbf{V}_t\). As the service times \(\delta^k_{\mathrm{u}}(n)\) and \(\delta^k_{\mathrm{d}}(n)\) are discrete quantities obtained by enumeration in~\eqref{Service time}, the expectation is not differentiable with respect to \(\mathbf{P}\) and \(\mathbf{V}_t\), preventing gradients backpropagation. 

To enable gradient-based updates, the next subsection derives a differentiable upper bound on this expectation.

\subsection{Differentiable Service-Time Expectation}
In this subsection, we derive a differentiable upper bound on \(E\big[e^{\theta_k \cdot \max\{\delta^k_{\mathrm{u}}(n),\delta^k_{\mathrm{d}}(n)\}}\big]\) as a closed-form function of \(\bar{\mathbf{P}}\) and \(\bar{\mathbf{V}}\), thereby enabling gradient backpropagation.

Denote the instantaneous UL and DL service times of AU$_k$ in the $t$-th coherence block as
\begin{equation}\label{eq:instantaneous service time}
\bar{\delta}^k_{\mathrm{u}}(t)=T_{\mathrm{c}}\Big\lceil \tfrac{L_{\mathrm{u}}}{T_{\mathrm{u}}R^{\mathrm{u}}_{k,n}(t)}\Big\rceil,
\quad
\bar{\delta}^k_{\mathrm{d}}(t)=T_{\mathrm{c}}\Big\lceil\tfrac{L_{\mathrm{d}}}{T_{\mathrm{d}}R^{\mathrm{d}}_{k,n}(t)}\Big\rceil,
\end{equation}

These quantities represent the time required to deliver an entire packet at the instantaneous rates achieved in coherence block \(t\). By substituting the UL and DL rates in~\eqref{UL capacity} and~\eqref{DL capacity}, respectively in~\eqref{eq:instantaneous service time}, both service times is closed-form functions of \(\mathbf{P}\) and \(\mathbf{V}_t\).

\begin{proposition}\label{proposition:service time upper bound}
When the rates are i.i.d. across coherence blocks, the service time expectation is upper bounded by
\begin{equation}\label{eq:upper bound of service time}
E_n\!\left[e^{\theta_k\max\{\delta^k_{\mathrm{u}}(n),\delta^k_{\mathrm{d}}(n)\}}\right]
\le E_t\!\left[e^{\theta_k\max\{\bar{\delta}^k_{\mathrm{u}}(t),\bar{\delta}^k_{\mathrm{d}}(t)\}}\right],
\end{equation}
where $\theta_k>0$, $\bar{\delta}^k_{\mathrm{u}}(t)$ and $\bar{\delta}^k_{\mathrm{d}}(t)$ denote the instantaneous UL and DL service times defined in~\eqref{eq:instantaneous service time}.
\end{proposition}
\begin{IEEEproof}
See Appendix~\ref{appendix:service time upper bound}.
\end{IEEEproof}

By replacing \(E_n\!\left[e^{\theta_k\max\{\delta^k_{\mathrm{u}}(n),\delta^k_{\mathrm{d}}(n)\}}\right]\)
in~\eqref{eq:Lagrangian DNNs} with its differentiable upper bound
\(E_t\!\left[e^{\theta_k\max\{\bar{\delta}^k_{\mathrm{u}}(t),\bar{\delta}^k_{\mathrm{d}}(t)\}}\right]\),
the required gradients with respect to \(\mathbf{P}\) and \(\mathbf{V}_t\) can be obtained through backpropagation, thus enabling gradient-based training.

\subsection{DNN Ddesign for UL and DL Policies}
In this subsection, we design DNNs for the UL and DL policies. For the UL policy, we first establish its PE property and design a GNN that preserves this property. For the DL policy, we exploit both its PE property and the optimal beamforming structure to design the DL GNN architecture.

\subsubsection{DNN Design for UL Policy} 
The UL policy exhibits a one-dimensional PE (1DPE)~property, as stated in the following proposition.
\begin{proposition}\label{proposition:PE of DL beamforming}
If $\mathbf{P}^{\star}$ is optimal for the input $\left(\boldsymbol{\alpha}, \boldsymbol{\theta}^{\star} \right)$, then $\boldsymbol{\Pi}_{K}^{\mathsf T}\mathbf{P}$ is optimal for the permuted input $\left(\boldsymbol{\Pi}_{K}^{\mathsf T}\boldsymbol{\alpha}, \boldsymbol{\Pi}_{K}^{\mathsf T}\boldsymbol{\theta}^{\star} \right)$, i.e.,
\begin{equation}\label{eq:1DPE of UL}
(\boldsymbol{\Pi}_{K}^{\mathsf T}\mathbf{P})=
\upphi_{\mathrm{d}}\left(\boldsymbol{\Pi}_{K}^{\mathsf T}\boldsymbol{\alpha}, \boldsymbol{\Pi}_{K}^{\mathsf T}\boldsymbol{\theta}^{\star} \right).
\end{equation}
where $\boldsymbol{\Pi}\!\in\!\mathbb{R}^{K\times K}$ is an arbitrary permutation matrix on the user~indices.
\end{proposition}
\begin{IEEEproof}
The proof follows from the invariance of the Lagrangian function~\eqref{Lagrangian} and constraints~\eqref{P3:power resction1} and~\eqref{P3:power resction2} under the corresponding index permutations and is omitted for brevity.
\end{IEEEproof}

To exploit the PE property in~\eqref{eq:1DPE of UL}, we implement the UL policy network $\Phi_{\mathrm{u}}\left(\boldsymbol{\alpha},\boldsymbol{\theta}^{\star};\boldsymbol{\omega}_{\mathbf{P}}\right)$ using a GNN defined on the fully connected graph, as shown in Fig.~\ref{GNN_graph}. This graph consists of $K$ user vertices with pairwise edges. For vertex $k$, the input feature is $(\alpha_k,\theta^{\star}_k)$, and the corresponding action is the UL power allocation $p_{k}$. The is no features or actions on the edges.

\begin{figure}[htbp]
\centering
\includegraphics[width=0.25\textwidth]{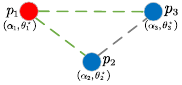}
\caption{Graph for UL policy with $K=3$.}  \label{GNN_graph}
\end{figure}

Denote $\mathbf{u}_{k}^{l}\in\mathbb{C}^{C_l\times1}$ as the hidden representation for user vertex $k$ in the $l$-th layer, with $C_l$ being the feature dimension. GNN iteratively updates it as  
\begin{equation}\label{eq:UL update equation}
\mathbf{u}_{k}^{l+1}=\sigma\Big(\mathbf{S}^l\mathbf{u}_{k}^{l}
\!+\!\mathbf{P}^l\!\!{\sum\limits_{i=1,\,i\neq k}^{M_{\mathrm{d}}}\!\!\mathbf{u}_{i,k}^{l}}\Big), 
\end{equation}
where $\mathbf{S}^{l}, \mathbf{P}^{l}\in\mathbb{C}^{C_{l+1}\times C_l}$ are trainable matrices, and $\sigma(\cdot)$ is the activation function. In an $L$-layer GNN, the input features $\mathbf{u}_{k}^{0}$ is the combination of $\alpha_{k}$ and $\theta_k$, and the output action $\mathbf{u}_{k}^{L}$ corresponds to the $k$-th column of $\mathbf{P}$, denoted as $\mathbf{p}_k\in\mathbb{C}^{M_{\mathrm{u}}\times1}$.

\subsubsection{DNN Design for DL Policy}
The DL policy exhibits a three-dimensional PE (3DPE)~property, as stated in the following proposition.
\begin{proposition}\label{proposition:PE of DL beamforming}
Let $\boldsymbol{\Pi}_{o}=\boldsymbol{\Pi}_{K}^{\mathsf T}
\otimes\boldsymbol{\Pi}_{N_t}^{\mathsf T}
\otimes\boldsymbol{\Pi}_{M_{\mathrm d}}^{\mathsf T}$. If $\mathbf{V}_t$ is optimal for input $(\mathbf{H}_t,\boldsymbol{\theta}^{\star})$, then $\mathrm{vec}(\hat{\mathbf{V}}_t)=\boldsymbol{\Pi}_{o}\mathrm{vec}\!\left(\mathbf{V}_t\right)$ is optimal for the permuted input $\mathrm{vec}(\hat{\mathbf{H}}_{t})=\boldsymbol{\Pi}_{o}\mathrm{vec}\!\left(\mathbf{H}_t\right)$ and $\hat{\boldsymbol{\theta}}^{\star}=\boldsymbol{\Pi}^{\mathsf{T}}_K\boldsymbol{\theta}^{\star}$,~i.e.,
\begin{equation}\label{3D PE}
\hat{\mathbf{V}}_{t}
=
\upphi_{\mathrm{d}}\big(\hat{\mathbf{H}}_{t},\hat{\boldsymbol{\theta}}^{\star}\big),
\end{equation}
where $\boldsymbol{\Pi}_{K}\in\mathbb{R}^{K\times K}$, $\boldsymbol{\Pi}_{N_t}\in\mathbb{R}^{N_t\times N_t}$, and $\boldsymbol{\Pi}_{M_{\mathrm d}}\in\mathbb{R}^{{M_{\mathrm d}}\times {M_{\mathrm d}}}$ are permutation matrices for users, antennas, and DL subchannels, respectively.
\end{proposition}
\begin{IEEEproof}
The proof follows from the invariance of the Lagrangian function~\eqref{Lagrangian} and constraint~\eqref{P3:precoder resction1} under the corresponding index permutations and is omitted for brevity.
\end{IEEEproof}

A direct approach to exploiting the 3DPE property is to use the multidimensional GNN proposed in~\cite{MDGNN}. However, this architecture operates on multidimensional input and output representations, resulting in a large model and high training complexity.
Prior work on narrowband communication systems has shown that a two-dimensional beamforming matrix can be represented using a one-dimensional power vector~\cite{Optimal_Structure}. Motivated by this result, we derive the optimal structure of the DL beamforming and incorporate it into the GNN architecture to reduce its output dimension.

\begin{proposition}
\label{prop:dl_precoder_structure}
If $\mathbf{V}_t=[\mathbf{v}_{1,1,t},\ldots,\mathbf{v}_{M_{\mathrm{d}},K,t}]\in\mathbb{C}^{M_{\mathrm{d}}\times N_t\times K}$ is the optimal beamforming matrix for input $\mathbf{H}_t$ of the problem~\eqref{P1:Power0}, the optimal beamforming vector $\mathbf{v}_{m,k,t}\in\mathbb{C}^{N_t\times 1}$ can be expressed as
\begin{equation}
\!\!\!\!\mathbf{v}_{m,k,t}\!=\!\sqrt{p_{m,k,t}^{\mathrm d}}\!
\frac{\big(\mathbf{I}_{N_t}\!\!+\!\!\sum\limits_{i=1}^{K}\frac{\eta_{i}\alpha_i}{\sigma_{\mathrm{d}}^2}
\mathbf{h}_{m,i,t}\mathbf{h}_{m,i,t}^{\mathsf H}\big)^{-1}\mathbf{h}_{m,k,t}}
{\big|\big(\mathbf{I}_{N_t}\!\!+\!\!\sum\limits_{i=1}^{K}\frac{\eta_{i}\alpha_i}{\sigma_{\mathrm{d}}^2}
\mathbf{h}_{m,i,t}\mathbf{h}_{m,i,t}^{\mathsf H}
\big)^{-1}
\mathbf{h}_{m,k,t}
\big|
},
\label{eq:optimal beamforming}
\end{equation}
where \(\eta _{i}\ge0\) denotes positive dual variables and \(p_{m,k,t}^{\mathrm d}\ge0\) is the DL transmit power allocated to
\(\mathrm{AU}_k\) on subchannel \(m\), satisfying $\sum_{m=1}^{M_{\mathrm{d}}}\sum_{k=1}^{K}p_{m,k,t}^{\mathrm d}
\leq P_{\max}^{\mathrm d}$.
\end{proposition}
\begin{IEEEproof}
See Appendix~\ref{appendix:dl_precoder_structure}.
\end{IEEEproof}

Proposition~\ref{prop:dl_precoder_structure} shows that the DL beamforming matrix can be recovered from the DL power allocation, and the dual variables through~\eqref{eq:optimal beamforming}. Therefore, instead of directly learning the high-dimensional matrix \(\mathbf{V}_t\), the DL policy network learns the reduced-dimensional policy
\begin{equation}\label{eq:DL power allocation policy}
(\mathbf{P}_t^{\mathrm d},\boldsymbol{\eta}_t)
=
\upphi_{\mathrm d}\left(\mathbf{H}_t,\boldsymbol{\theta}^{\star}\right),
\end{equation}
where $\mathbf{P}_t^{\mathrm d}=[p_{1,1,t}^{\mathrm d},\cdots,p_{M_{\mathrm{d}},K,t}^{\mathrm d}]\in\mathbb{R}^{M_{\mathrm{d}}\times K}$
and \(\boldsymbol{\eta}_t=[\eta_{1,t},\ldots,\eta_{K,t}]^{\mathsf T}
\in\mathbb{R}^{K\times1}\)
denote the optimal DL power allocation matrix and the dual variable vector associated with the input $\left( \mathbf{H}_{t}, \boldsymbol{\theta}^{\star} \right)$. 

The DL policy~\eqref{eq:2D PE of DL} exhibits a PE property along user and subchannel dimensions.  It follows from~\eqref{eq:optimal beamforming} that jointly permuting the user indices of \(\mathbf{H}_t\), \(\mathbf{P}_t^{\mathrm d}\), and \(\boldsymbol{\eta}_t\) induces the corresponding user permutation in \(\mathbf{V}_t\). Similarly, jointly permuting the DL subchannel indices of \(\mathbf{H}_t\) and \(\mathbf{P}_t^{\mathrm d}\) induces the corresponding subchannel permutation in \(\mathbf{V}_t\). As established in Proposition~\ref{proposition:PE of DL beamforming}, the resulting beamforming matrix remains optimal for the correspondingly permuted inputs \(\mathbf{H}_t\) and \(\boldsymbol{\theta}^{\star}\). Consequently, the reduced DL policy exhibits a two-dimensional permutation-equivariance (2DPE) property, expressed as
\begin{equation}\label{eq:2D PE of DL}
\begin{aligned}
&\left(
\boldsymbol{\Pi}^{\mathsf T}_{M_{\mathrm d}}
\mathbf{P}_t^{\mathrm d}\boldsymbol{\Pi}_{K},
\boldsymbol{\Pi}^{\mathsf T}_{K}\boldsymbol{\eta}_t
\right)\\
&\quad=
\upphi_{\mathrm d}\left(
\left(
\boldsymbol{\Pi}_{K}^{\mathsf T}
\otimes\mathbf{I}_{N_t}
\otimes\boldsymbol{\Pi}_{M_{\mathrm d}}^{\mathsf T}
\right)
\mathrm{vec}(\mathbf{H}_t),
\boldsymbol{\Pi}^{\mathsf T}_{K}\boldsymbol{\theta}^{\star}
\right).
\end{aligned}
\end{equation}

To exploit the 2DPE property, we implement the DL policy using a GNN defined on the bipartite graph shown in Fig.~\ref{GNN_graph}. The graph contains two types of vertices, namely DL subchannel vertices and user vertices, connected with edges. Features and actions are associated with the edges and user vertices. Specifically, edge \((m,k)\) takes the composite channel vector \(\sqrt{\alpha_k}\mathbf{h}_{m,k,t}\) as its feature and the corresponding DL power allocation \(p_{m,k,t}^{\mathrm d}\) as its action. User vertex \(k\) takes \(\theta_k^{\star}\) as its feature and \(\eta_{k,t}\) as its action. The subchannel vertices have neither features nor actions.

\begin{figure}[htbp]
\centering
\includegraphics[width=0.35\textwidth]{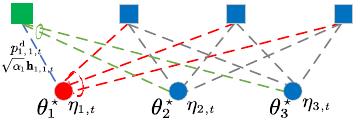}
\caption{Graph for DL policy with $M_{\mathrm{d}}=4$ and $K=3$.}  \label{GNN_graph}
\end{figure}

 Let $\mathbf{d}_{m,k}^{l}\in\mathbb{C}^{C_l\times1}$ and $\mathbf{a}_{k}^{l}\in\mathbb{C}^{C_l\times1}$ denote the hidden representations of edge $(m,k)$ and user vertex $k$ at layer $l$, respectively. These representations are updated as
\begin{subequations}\label{update_equation_of_GNN}
    \begin{align}
\mathbf{d}_{m,k}^{l+1}
&=\sigma\Big(\mathbf{S}^l\mathbf{d}_{m,k}^{l}
\!+\!\mathbf{P}_1^l\!\!\!\!{\sum\limits_{i=1,\,i\neq m}^{M_{\mathrm{d}}}\!\!\mathbf{d}_{,i,k}^{l}}+\mathbf{P}_2^l\!\!\!\!{\sum\limits_{j=1,\,j\neq k}^{K}\!\!\mathbf{d}_{n,j}^{l}}\Big), \label{edge update equation}\\
\mathbf{a}_{k}^{l+1}&=\sigma\Big(\mathbf{Q}_{1}^{l}\mathbf{a}_{k}^{l}
+\mathbf{Q}_{2}^{l}\sum\nolimits_{i=1}^{M_{\mathrm{d}}}\mathbf{d}_{i,k}^{l}\Big), \label{vertex update equation}
    \end{align}
\end{subequations}
where $\mathbf{S}^{l}, \mathbf{P}_{m}^{l}, \mathbf{Q}_m^{l}\in\mathbb{C}^{C_{l+1}\times C_l}, m\in\{1,2\}$ are trainable matrices. In an \(L\)-layer GNN, the edge and user representations are initialized using
\(\mathbf{d}_{m,k}^{0}=\sqrt{\alpha_k}\mathbf{h}_{m,k,t}\) and
\(\mathbf{a}_{k}^{0}=\theta_k^{\star}\), respectively. The final edge representation \(\mathbf{d}_{m,k}^{L}\) produces \(p_{m,k,t}^{\mathrm d}\), whereas \(\eta_{k,t}\) is obtained by averaging the elements of \(\mathbf{a}_{k}^{L}\).

\section{Simulation Results}
In this section, we evaluate the performance of the proposed method by comparing with relative~baselines.

\subsection{Simulation Setup}
Unless otherwise specified, all simulations are conducted using the following setup. Let \(N_0\) denote the single-sided noise power spectral density, and let \(N_{\mathrm F}^{\mathrm u}\) and \(N_{\mathrm F}^{\mathrm d}\) denote the receiver noise figures for the UL and DL, respectively. The noise power over each subchannel is $\sigma_x^2=W_0\,10^{(N_0-30)/10}10^{N_{\mathrm F}^{x}/10}$ with $x\in\{\mathrm u,\mathrm d\}$. The normalized UL SNR of \(\mathrm{AU}_k\) is defined as $\mathrm{SNR}_k
=
10\log_{10}\big(\frac{\alpha_k}{\sigma_{\mathrm u}^2}\big)\,\mathrm{dB}$. To characterize heterogeneous large-scale channel conditions, we take \(\mathrm{AU}_1\) as the reference user and express $\mathrm{SNR}_k=\mathrm{SNR}_1+\Delta\mathrm{SNR}_{k,1}$ where \(\Delta\mathrm{SNR}_{k,1}\) denotes the SNR difference between \(\mathrm{AU}_k\) and \(\mathrm{AU}_1\). The large-scale channel gain of \(\mathrm{AU}_k\) is obtained as $\alpha_k=\sigma_{\mathrm u}^2 10^{\mathrm{SNR}_k/10}$, and the small-scale channel gains follow Rayleigh fading. The remaining simulation parameters and hyperparameters are listed in Table~\ref{parameter}.

\begin{table} [htbp]
	\centering
	\caption{Simulation Parameters and Hyper-parameters}
	\begin{tabular}{c|c}
		\hline
        \text{Number of AUs} $K$  & $2$ \\
		\hline
		\text{Target PLR} $\varepsilon_{\max}^1, \varepsilon_{\max}^2$  & $10^{-2}, 10^{-3}$ \\
		\hline
		\text{PDB} $D_{\max}^1, D_{\max}^2$  &  20, 20 ms \\
		\hline
		\text{Duration of coherence block} $T_c$ & 1 ms \\
		\hline
		\text{Duration of transmission} $T_{\mathrm{u}}, T_{\mathrm{d}}$ & 0.5, 0.5 ms \\
		\hline
		\text{Frame size} $L_{\mathrm{u}}$, $L_{\mathrm{d}}$ & $2$, $100$ Kb  \\
		\hline
		\text{Frame rate} $f$ & 120 fps  \\
		\hline
		\multirow{2}{*}{Truncated Gaussian distributed}  & $\left(\mu, \sigma, b_1, b_2\right)$\\
		arrival process & $(\frac{1}{f}, 2, \frac{1}{f}-5, \frac{1}{f}+5)$\\
		\hline
		\text{Maximum transmit power} $P_{\max}^{\mathrm{u}},P_{\max}^{\mathrm{d}}$  & 23, 46 dBm \\
		\hline
        \text{Weighting coefficient} $\beta$ &$ 5\times 10^{-3}$ \\
		\hline
		\text{Number of antennas} $N_{t}$  & 8 \\
		\hline
		\text{Number of subchannels} $M_{\mathrm{u}}$, $M_{\mathrm{d}}$ & 11, 24 (5MHz, 10MHz BW)\\
		\hline
		\text{Bandwidth of subchannels} $W_0$  & 360 kHz \\
		\hline
		\text{Single-sided noise spectral density} $N_{0}$  & -174 dBm/Hz \\
		\hline
		\text{Noise figure} $N_{\mathrm{F}}^{\mathrm{u}}, N_{\mathrm{F}}^{\mathrm{d}}$  & 5, 3 dB \\
        \hline
        \text{SNR for} $\mathrm{AU}_1$ $\mathrm{SNR}_{1}$  & 0 dB \\
		\hline
	\end{tabular}
	\label{parameter}
\end{table}

Both the UL and DL GNNs consist of six hidden layers with widths \(\{64,128,512,512,128,64\}\), and a LeakyReLU activation function is applied after each hidden layer. The two networks are jointly trained in an unsupervised manner using the loss function in~\eqref{eq:Lagrangian DNNs}. The backpropagation gradients are computed according to~\eqref{eq:Gradient of DNNs} and~\eqref{eq:gradients of expection}. Since the ceiling operation in~\eqref{eq:instantaneous service time} is non-differentiable, we adopt a straight-through estimator by setting its backward gradient to one, following~\cite{Multi_Domain_Correlation}. We generate \(500{,}000\) samples for training and an independent set of \(50{,}000\) samples for testing. The two GNNs are trained using the Adam optimizer with initial learning rates \(\xi_{\bar{\mathbf{P}}}=\xi_{\bar{\mathbf{V}}}=10^{-3}\) and batch size \(|\mathcal{B}|=64\).

\subsection{Performance Evaluation}
\subsubsection{Theoretical PDB Violation Probability}
We first evaluate the tightness of the derived upper bound of PDB violation probability by comparing it with existing upper bounds derived in~\cite{Learn_to_Optimize,A_combined,Dealing_wit}, where the inter-arrival times of data packets follow a truncated Gaussian distribution, as specified in Table~\ref{parameter} and the service times for the two nodes follow an exponential distribution with an average service time of 4 slots. Fig.~\ref{tightness} shows that the derived upper bound closely approximates the simulated results in the region $D_{\max}\in[15,25]$, demonstrating the tightness of the proposed~bound.

\begin{figure}[htbp]
\centering
\includegraphics[width=0.5\textwidth]{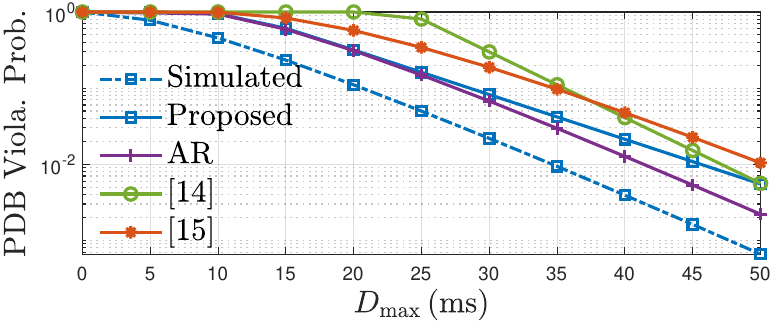}
\caption{Tightness of bounds for PDB violation probability. }  \label{tightness}
\end{figure}

\subsubsection{Performance Comparison with Baselines}
To evaluate the performance gain of the proposed method, we compare it with the following baselines:
\begin{itemize}
	\item[$\bullet$] \textbf{Time-Const.:} Following the assumption in~\cite{MTG1}, each packet has equal and constant UL and DL service times, determined by~\eqref{eq:qos upper bound}. The UL transmit power of each AU is obtained by a bisection search such each packer can be delivered within the prescribed constant service time. For the DL, We further assume that the bits of each packet are uniformly transmitted over the available slots. In each slot  the minimum mean-square	error (WMMSE) algorithm is used to optimized the beamforming vectors for transmit-power minimization~\cite{Optimal_Structure}.
		
	\item[$\bullet$] \textbf{UL-DL:} This method from~\cite{Learn_to_Optimize}, where the QoS exponent \(\theta_k^\star\) is obtained by using the theoretical PDB upper bound derived in~\cite{Learn_to_Optimize}. Given the obtained \(\theta_k^\star\), the original joint optimization problem is decomposed into separate UL and DL subproblems. The corresponding UL and DL policy networks are then trained independently in an unsupervised manner.
	
	\item[$\bullet$] \textbf{Upp-Bou.:} This method follows the same procedure with UL-DL, except it computes the QoS exponent \(\theta_k^\star\) of \(\mathrm{AU}_k\) using the theoretical upper bound derived in~\cite{Time_domain}.

    \item[$\bullet$] \textbf{Two-Queue:} Following~\cite{Two_error}, this method decomposes the tandem queue into two single-server queues corresponding to the UL and DL, respectively. Their QoS requirements are denoted by \((D_{\mathrm{u}}^{k},\varepsilon_{\mathrm{u}}^{k})\) and \((D_{\mathrm{d}}^{k},\varepsilon_{\mathrm{d}}^{k})\), where the end-to-end delay budget is equally divided as \(D_{\mathrm{u}}^{k}=D_{\mathrm{d}}^{k}=D_{\max}^k/2\). The violation-probability targets are also set equal, i.e., \(\varepsilon_{\mathrm{u}}^{k}=\varepsilon_{\mathrm{d}}^{k}\), and are chosen to satisfy $1-\left(1-\varepsilon_{\mathrm{u}}^{k}\right)\left(1-\varepsilon_{\mathrm{d}}^{k}\right)=\varepsilon_{\max}^k$. The UL and DL policy networks are subsequently trained independently to satisfy the QoS requirement of each queue.
\end{itemize}

Fig.~\ref{Power gain SNR} shows the power saving gains under different SNRs. The performance metric is defined as the percentage reduction in weighted UL–DL power consumption achieved by different methods relative to the Two-Queue baseline. As Two-Queue is the reference and its power saving gain is zero by definition, it is not
shown as a separate curve in this figure. The results show that all the methods constantly achieves positive power saving gains relative to Two-Queue, because Two-Queu decomposes the tandem system into two single-server queues neglecting the statistical coupling between the UL and DL queues. The proposed method achieves the greatest power saving gains, because it jointly optimizes the UL and DL resource-allocation policies, thereby allowing their service-time distributions to adapt to the packet-arrival statistics. Moreover, UL-DL outperforms Upp-Bou. because it employs a tighter PDB violation probability bound than Upp-Bou.

\begin{figure}[htbp]
	\centering
	\includegraphics[width=.9\linewidth]{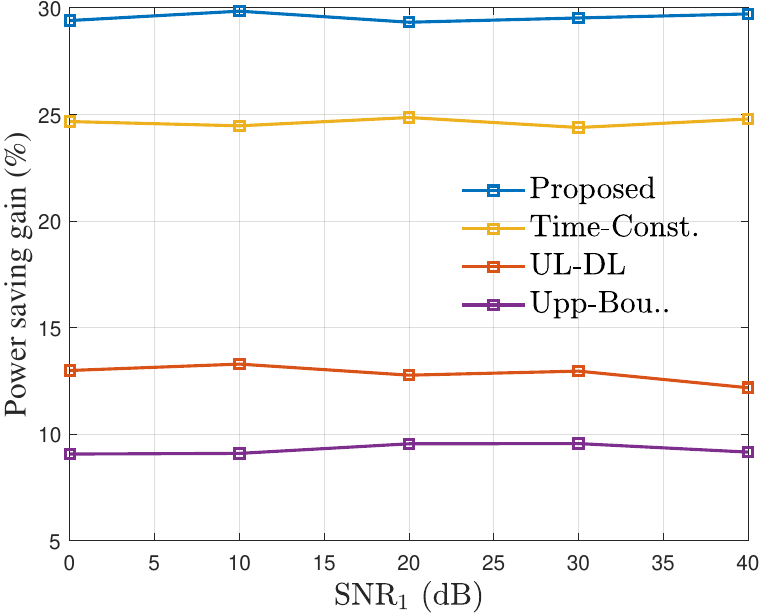}
	\caption{Power saving gains under different SNRs with $\mathrm{SNR}^{\varDelta}_{2,1}=5\,$dB.}
	\label{Power gain SNR}
\end{figure}

Fig.~\ref{Power gain beta} compares the power saving gains under different
weighting coefficients \(\beta\). The proposed method consistently achieves the greatest power saving gains because it jointly adapts the UL and DL transmit power according to the values of $\beta$. Its advantage becomes more pronounced as \(\beta\) increases, because the achievable power saving gains in DL (i.e. $\beta=1$) is much larger than that in UL (i.e. $\beta=0$) and  a larger \(\beta\) assigns greater importance to the DL. Among the baselines, UL-DL outperforms Upp-Bou. because it employs a tighter PDB violation probability bound and therefore yields a less conservative resource allocation. At \(\beta=0\), Time-Const. is the least effective methods because it adopts a conservative UL policy that ensures that every packet is delivered within the prescribed constant service time. At
\(\beta=1\), however, Time-Const. outperforms UL-DL and Upp-Bou., primarily
because the WMMSE algorithm employed in its DL policy can efficiently obtain
a high-quality locally optimal~solution.
\begin{figure}[htbp]  
	\centering
	\includegraphics[width=.9\linewidth]{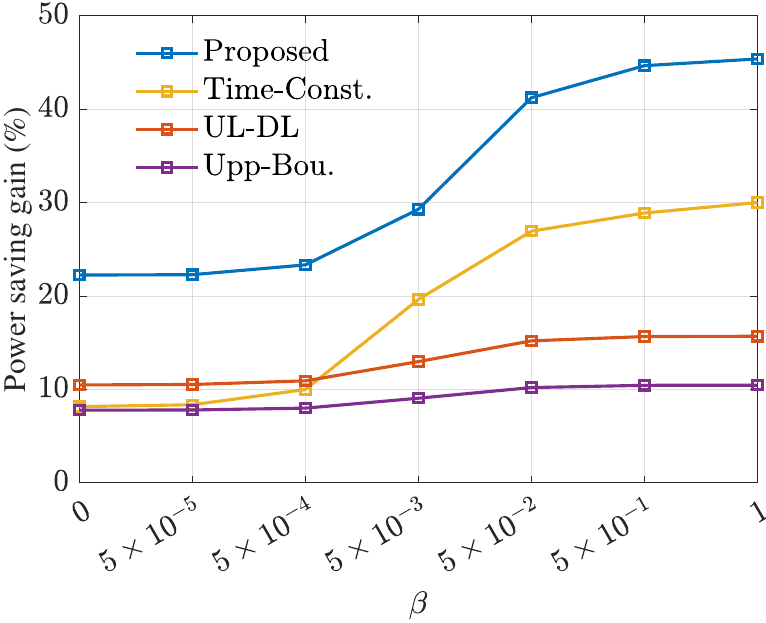}
	\caption{Power saving gains under different $\beta$ with $\mathrm{SNR}^{\varDelta}_{2,1}=5\,$dB.}
	\label{Power gain beta}
\end{figure}

Fig.~\ref{Total Service time comparison} shows the average total service time,
defined as the sum of the average UL and DL service times. The proposed method
achieves a higher average total service time than UL-DL, Upp-Bou., and
Two-Queue. This result partly explains its higher power-saving gain, because a
longer service time generally permits transmission at a lower rate and power.
Although, Time-Const. achieves a comparable or even longer total service time
than the proposed method, it yields a smaller power-saving gain because its
fixed service times prevent the UL and DL service processes from adapting to
the channel and packet-arrival statistics.

\begin{figure}[htbp]   
	\centering
	\includegraphics[width=.9\linewidth]{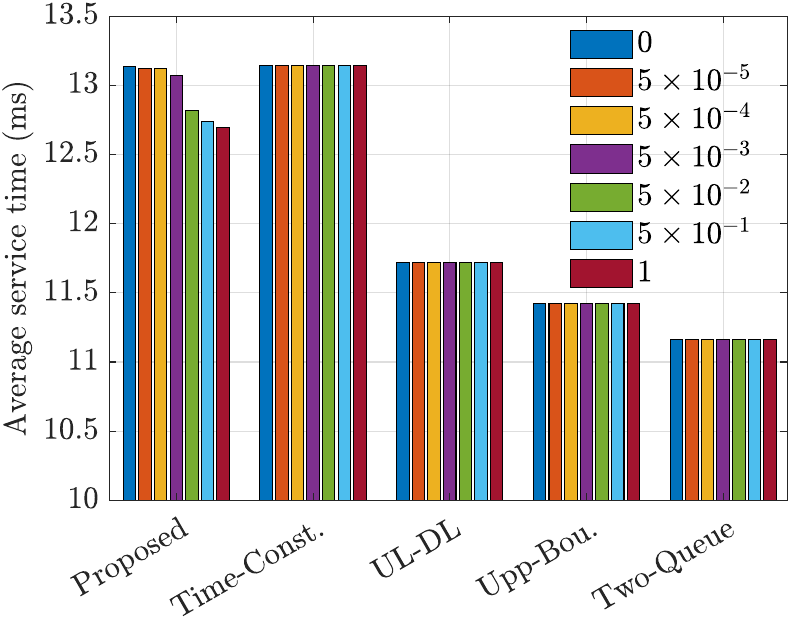}
	\caption{Average total service time under different $\beta$ with  $\mathrm{SNR}^{\varDelta}_{2,1}=5\,$dB.}
	\label{Total Service time comparison}
\end{figure}

Fig.~\ref{Service time comparison} presents the average UL and DL service
times for different values of \(\beta\). As shown in
Figs.~\ref{UL service time} and~\ref{DL service time}, respectively, the
proposed method decreases the average UL service time and increases the
average DL service time as \(\beta\) increases. This adjustment increases the
UL transmit power while reducing the DL transmit power. In contrast,
the service times of the baselines vary little with \(\beta\), limiting their
ability to adapt the UL--DL power tradeoff.

\begin{figure}[htbp]
	\centering
	\begin{minipage}[t]{0.9\linewidth}	
		\subfigure[Service time of UL.]{
			\includegraphics[width=\textwidth]{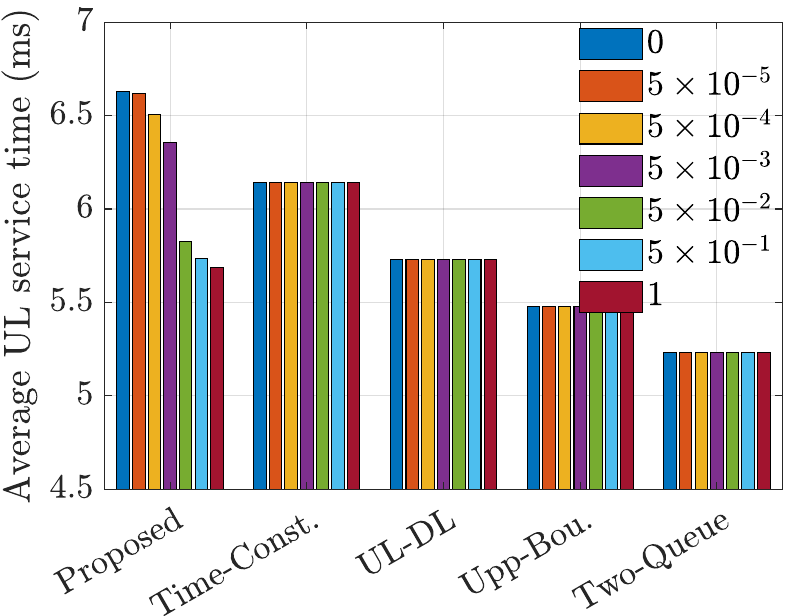}\label{UL service time}}
	\end{minipage}
	\begin{minipage}[t]{0.9\linewidth}	
		\subfigure[Service time of DL.]{
			\includegraphics[width=\textwidth]{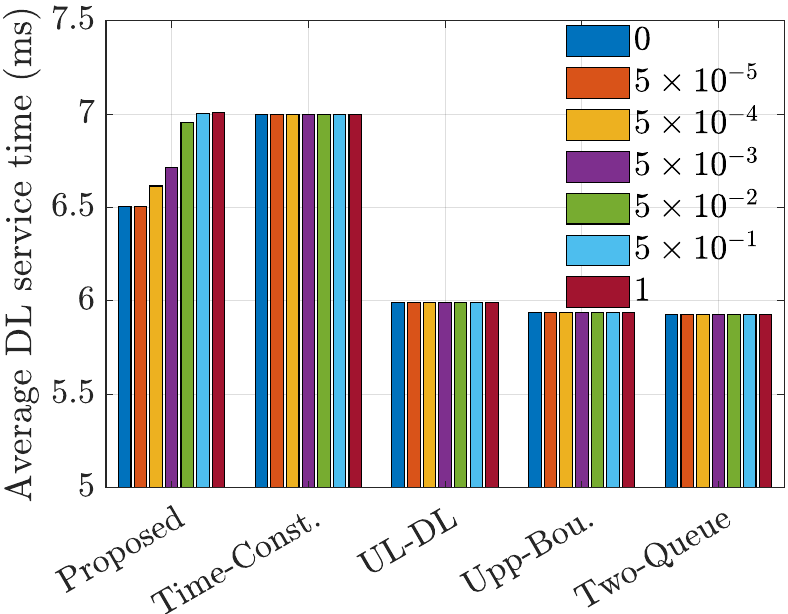}\label{DL service time}}
	\end{minipage}
	\caption{Service time comparison under different $\beta$ with $\mathrm{SNR}^{\varDelta}_{2,1}=5\,$dB.}
	\label{Service time comparison}
\end{figure}

\begin{table*}[htbp]  
	\centering
	\caption{Evaluating the Reliability of the Proposed Method}
	\label{Achieved PDB violation probability}
	\begin{tabular}{c|c|c|c|c|c}
		\hline
		\hline
		\text{Weight coefficient $\beta$}                                & $5\times 10^{-5}$   & $5\times 10^{-4}$    & $5\times 10^{-3}$    & $5\times 10^{-2}$    & $5\times 10^{-1}$ \\
		\hline
		\text{Achieved PDB violation probability of \( \mathrm{AU}_1 \)} & $2.60\times 10^{-3}$ & $2.60\times 10^{-3}$ & $1.57\times 10^{-3}$ & $1.30\times 10^{-3}$ & $1.57\times 10^{-3}$ \\
		\hline
		\text{Achieved PDB violation probability of \( \mathrm{AU}_2 \)} & $2.70\times 10^{-4}$ & $2.70\times 10^{-4}$ & $4.10\times 10^{-4}$ & $1.40\times 10^{-4}$ & $2.20\times 10^{-4}$ \\
		\hline
		\hline
	\end{tabular}
	\vspace{-0.5cm}
\end{table*}

Table~\ref{Achieved PDB violation probability} reports the empirical PDB
violation probabilities achieved by the proposed method. For each AU, the probability
is estimated from \(10^8\) Monte Carlo trials. In each trial, the E2E packet delays are recorded and the PDB violation probability is estimated as the fraction of packets whose delays exceed the corresponding
PDB threshold. All reported probabilities remain below their prescribed
thresholds, confirming that the proposed method satisfies the reliability requirements under the considered simulation settings.

\section{Conclusion}
This paper investigated joint UL and DL resource optimization for MEC-assisted wireless AR systems under PDB violation probability constraints. We modeled the AR transmission process as a tandem queueing system and derived an upper bound on the PDB violation probability using SNC and Doob's inequality. Based on this bound, the original probabilistic QoS constraint was transformed into a tractable service-time condition that captures the joint behavior of the UL and DL service processes. We then formulated a weighted UL-DL transmit-power minimization problem and designed a learning framework to jointly optimize UL power allocation and DL beamforming. To reduce training complexity, we further developed GNN-based policies for UL power allocation and DL beamforming by exploiting PE properties and the optimal solution structure.
Simulation results verified the tightness of the derived PDB bound and showed that the proposed method satisfies the AR reliability requirement while reducing the weighted transmit power compared with baselines that optimize UL and DL resources separately.

\appendices
\numberwithin{equation}{section}  
\section{Proof of Proposition~\ref{proposition:service time upper bound}} 
\label{appendix:service time upper bound}

We prove this proposition by Jensen's inequality. We first introduce the definition of Jensen's inequality.

\begin{definition}[Jensen's inequality]\label{Jensen's Inequality}
Let $X$ be a discrete random variable taking positive values
$x_1,\ldots,x_N$ with probabilities $p_1,\ldots,p_N$, where
$\sum_{i=1}^{N}p_i=1$. For any convex function $g(\cdot)$, we have
\begin{equation}\label{eq:jensen inequality}
     g\big(\sum_{i=1}^{N}p_i x_i\big)
    \leq \sum_{i=1}^{N}p_i g(x_i).
\end{equation}
\end{definition}

Denote
$\delta_k^{\max}(n)=\max\{\delta^k_{\mathrm{u}}(n),
\delta^k_{\mathrm{d}}(n)\}$. Without loss of generality, assume that the
UL determines the service time of the $n$-th packet, i.e.,
$\delta_k^{\max}(n)=\delta^k_{\mathrm{u}}(n)$ (the DL-dominant case is symmetric). Let $\mathcal{T}_k(n)$ denote the set of coherence blocks used to transmit the $n$-th packet, and $l_k^{\max}(n)=\delta_k^{\max}(n)/T_{\mathrm{c}}$ be the number of such blocks. The
average UL rate over these blocks is given by
\begin{equation}\label{eq:everage UL rate}
    \bar{R}_{k,n}^{\max}
    =\frac{1}{l_k^{\max}(n)}
    \sum_{t\in\mathcal{T}_k(n)} R^{\mathrm{u}}_{k,t}.
\end{equation}
Using~\eqref{eq:everage UL rate}, $\delta_k^{\max}(n)$ can be expressed in
terms of the average rate as
\begin{equation}\label{eq:service time with average service time}
    \delta_k^{\max}(n)
    =\frac{T_{\mathrm{c}}L_{\mathrm{u}}}
    {T_{\mathrm{u}}\bar{R}_{k,n}^{\max}}.
\end{equation}

Defining $g_k(x)\triangleq\exp\!\big(\theta_k\frac{T_{\mathrm{c}}L_{\mathrm{u}}}{T_{\mathrm{u}}x}\big)$, since $1/x$ is convex over $x>0$ and the exponential function is convex and increasing, $g_k(x)$ is convex. Applying Jensen's inequality with uniform
weights $1/l_k^{\max}(n)$ gives
\begin{equation}\label{Jensen bound}
\begin{aligned}
e^{\theta_k\delta_k^{\max}(n)}
&\leq
g_k\!\big(
\frac{1}{l_k^{\max}(n)}
\sum_{t\in\mathcal{T}_k(n)} R^{\mathrm{u}}_{k,t}
\big)  \\
&=
\frac{1}{l_k^{\max}(n)}
\sum_{t\in\mathcal{T}_k(n)}
e^{\theta_k\frac{T_{\mathrm{c}}L_{\mathrm{u}}}{T_{\mathrm{u}}R^{\mathrm{u}}_{k,n}}}\\
&\leq
\frac{1}{l_k^{\max}(n)}
\sum_{t\in\mathcal{T}_k(n)}
e^{\theta_k\bar{\delta}^k_{\mathrm{u}}(t)}.
\end{aligned}
\end{equation}
where the first inequality holds as the Jensen's inequality and the last inequality holds as the ceiling operation in the instantaneous service time, as defined
in~\eqref{eq:instantaneous service time}.

Since $e^{\theta_k x}$ is monotonically increasing in $x$, it follows that
\begin{equation}\label{eq:add one instance service time}
e^{\theta_k\bar{\delta}^k_{\mathrm{u}}(t)}
\leq
e^{\theta_k\max\{
\bar{\delta}^k_{\mathrm{u}}(t),
\bar{\delta}^k_{\mathrm{d}}(t)\}} .
\end{equation}
Substituting~\eqref{eq:add one instance service time} into~\eqref{Jensen bound} yields the per-packet
upper bound
\begin{equation}\label{per-packet bound}
\!\!\!e^{\theta_k\max\{\delta^k_{\mathrm{u}}(n),
\delta^k_{\mathrm{d}}(n)\}}
\!\leq\!
\frac{1}{l_k^{\max}(n)}\!\!\!\sum_{t\in\mathcal{T}_k(n)}
\!\!\!\!e^{\theta_k\max\{
\bar{\delta}^k_{\mathrm{u}}(t),
\bar{\delta}^k_{\mathrm{d}}(t)\}} .
\end{equation}
As the rates across coherence blocks are i.i.d., taking expectations on both sides of~\eqref{per-packet bound} yields
\begin{equation}\label{eq:upper bound of service time in appendix}
E_n\left[
e^{\theta_k\max\{\delta^k_{\mathrm{u}}(n),
\delta^k_{\mathrm{d}}(n)\}}
\right]
\leq
E_t\left[
e^{\theta_k\max\{
\bar{\delta}^k_{\mathrm{u}}(t),
\bar{\delta}^k_{\mathrm{d}}(t)\}}
\right], 
\end{equation}
which complete the proof.

\section{Proof of Proposition~\ref{prop:dl_precoder_structure}}
\label{appendix:dl_precoder_structure}
We relate the optimal DL beamforming solution of the reformulated problem
\(\mathrm{P}1\) to an auxiliary transmit-power minimization problem with
minimum-rate constraints. Let \(\mathbf{V}_t^{\star}\) denote the optimal DL
beamforming matrix in coherence block \(t\), and let
\(R_{k,t}^{\mathrm d,\star}\) denote the corresponding DL rate of
\(\mathrm{AU}_k\). The auxiliary problem is formulated as
\begin{subequations}
\label{P2}
\begin{align}
\mathrm{P}2:
\min_{\{\mathbf{v}_{m,k,t}\}}
&\sum_{m=1}^{M_{\mathrm{d}}}\sum_{k=1}^{K}|\mathbf{v}_{m,k,t}|^2
\label{P2:obj}\\
\mathrm{s.t.}&\sum_{m=1}^{M_{\mathrm{d}}} W_0\log_2\left(1+\gamma_{m,k,t}^{\mathrm d}\right)
\geq R_{k,t}^{\mathrm d,\star}, \,\, \forall k .
\label{P2:rate}
\end{align}
\end{subequations}
We first show by contradiction that \(\mathbf{V}_t^{\star}\) is also optimal
for \(\mathrm{P}2\). Suppose otherwise. Then, there exists another beamforming
matrix \(\mathbf{V}_t'\) satisfying~\eqref{P2:rate} such that
\(|\mathbf{V}_t'|^2
<|\mathbf{V}_t^{\star}|^2\).
Because \(\mathbf{V}_t'\) achieves DL rates no smaller than
\(R_{k,t}^{\mathrm d,\star}\), it does not increase any DL service time and
thus preserves QoS feasibility. Replacing \(\mathbf{V}_t^{\star}\) with
\(\mathbf{V}_t'\) would strictly reduce the DL power term in the objective of
\(\mathrm{P}1\), contradicting the optimality of~\(\mathbf{V}_t^{\star}\).

We next characterize the optimal beamforming structure using the first-order
stationarity conditions of the auxiliary problem. Substituting~\eqref{DL_SINR}
into~\eqref{P2:rate} gives
\begin{equation}
\!\!1\!+\!\frac{\alpha_k\left|\mathbf{h}_{m,k,t}^{\mathsf H}\mathbf{v}_{m,k,t}\right|^2}{\sum\limits_{j=1,j\ne k}^{K}\!\!\!\!
\alpha_k\left|\mathbf{h}_{m,k,t}^{\mathsf H}\mathbf{v}_{m,j,t}\right|^2
 +\sigma_{\mathrm{d}}^2}\geq
\frac{\gamma^{\mathrm{th}}_k}{\prod\nolimits_{i=1}^{M_{\mathrm{d}}}(1+\gamma^{\mathrm{d}}_k)},
\label{eq:fixed SINR constraint}
\end{equation}
where $\gamma^{\mathrm{th}}_k=e^{R_{k,t}^{\mathrm d,\star}/W_0}$ denotes the rate induced threshold. Define $\xi^{\mathrm{th}}_{k,m,t}=\frac{\gamma^{\mathrm{th}}_k}{\prod\nolimits_{i=1}^{M_{\mathrm{d}}}(1+\gamma^{\mathrm{d}}_k)}-1$ as the corresponding SINR threshold. The
constraint can then be written as
\begin{equation}
\frac{\alpha_k\left|\mathbf{h}_{m,k,t}^{\mathsf H}\mathbf{v}_{m,k,t}\right|^2}{\xi^{\mathrm{th}}_{k,m,t}\sigma_{\mathrm{d}}^2}\geq\sum\limits_{j=1,j\ne k}^{K}\!
\!\!\frac{\alpha_k}{\sigma_{\mathrm{d}}^2}\left|\mathbf{h}_{m,k,t}^{\mathsf H}\mathbf{v}_{m,j,t}\right|^2+1.
\label{eq:fixed SINR constraint 1}
\end{equation}
Replacing the constraint \eqref{P2:rate} with \eqref{eq:fixed SINR constraint 1}, the corresponding Lagrangian is
\begin{equation}
\begin{split}\raisetag{2.5cm}
&\mathcal{L}
=\sum\limits_{m=1}^{M_{\mathrm{d}}}\sum\limits_{k=1}^{K}|\mathbf{v}_{m,k,t}|^2 +\\
&
\sum\limits_{k=1}^{K}\lambda_{k}\Big(1+\!\!\!\!\sum\limits_{j=1,j\ne k}^{K}\!
\!\!\frac{\alpha_k}{\sigma_{\mathrm{d}}^2}\left|\mathbf{h}_{m,k,t}^{\mathsf H}\mathbf{v}_{m,j,t}\right|^2\!\!-\frac{\alpha_k\left|\mathbf{h}_{m,k,t}^{\mathsf H}\mathbf{v}_{m,k,t}\right|^2}{\xi^{\mathrm{th}}_{k,m,t}\sigma_{\mathrm{d}}^2}\Big).
\end{split}
\label{eq:lagrangian P2}
\end{equation}

Applying the first-order stationarity condition, i.e., \(\partial\mathcal{L}/\partial\mathbf{v}_{m,k,t}=\mathbf{0}\) yields
\begin{equation}
\begin{aligned}
&\mathbf{v}_{m,k,t}+
\sum\limits_{j=1,j\ne k}^{K}
\frac{\lambda_{j}\alpha_j}{\sigma_{\mathrm{d}}^2}
\mathbf{h}_{m,j,t}\mathbf{h}_{m,j,t}^{\mathsf H}\mathbf{v}_{m,k,t}=\\
&\quad\quad\quad\quad\quad\quad\quad\quad\quad\quad\frac{\lambda_{k}\alpha_k}
{\xi_{m,k,t}\sigma_{\mathrm{d}}^2}
\mathbf{h}_{m,k,t}\mathbf{h}_{m,k,t}^{\mathsf H}\mathbf{v}_{m,k,t}.
\end{aligned}
\label{eq:kkt1}
\end{equation}
By adding $\frac{\lambda_{k}\alpha_k}{\sigma_{\mathrm{d}}^2}
\mathbf{h}_{m,k,t}\mathbf{h}_{m,k,t}^{\mathsf H}\mathbf{v}_{m,k,t}$ to the both sides of \eqref{eq:kkt1} and rearranging the resulting terms gives
\begin{equation}
\begin{aligned}
\mathbf{v}_{m,k,t}
=&\Big(\mathbf{I}_{N_t}+\sum_{j=1}^{K}
\frac{\lambda_{j}\alpha_j}{\sigma_{\mathrm{d}}^2}
\mathbf{h}_{m,j,t}\mathbf{h}_{m,j,t}^{\mathsf H}
\Big)^{-1}
\mathbf{h}_{m,k,t} \\
&\times
\underbrace{
\frac{\lambda_{m}\alpha_k}{\sigma_{\mathrm{d}}^2}
\Big(1+\frac{1}{\xi_{m,k,t}}\Big)
\mathbf{h}_{m,k,t}^{\mathsf H}\mathbf{v}_{m,k,t}
}_{\mathrm{scalar}} .
\end{aligned}
\label{eq:kkt3}
\end{equation}
Equation~\eqref{eq:kkt3} shows that the optimal beamforming vector is
collinear with the first term in~\eqref{eq:kkt3} because the remaining multiplicative term is a scalar. Therefore, each optimal beamforming vector can be decomposed into an allocated power and a
unit-norm beamforming direction as
$\mathbf{v}_{m,k,t}^{\star}\in\mathbb{C}^{N_t\times 1}$ can be expressed as
\begin{equation}
\!\!\!\!\!\!\mathbf{v}_{m,k,t}\!=\!\sqrt{p_{m,k,t}^{\mathrm d}}\!
\frac{\!\big(\mathbf{I}_{N_t}\!\!+\!\!\sum\limits_{i=1}^{K}\frac{\eta_{i,t}\alpha_i}{\sigma_{\mathrm{d}}^2}
\mathbf{h}_{m,i,t}\mathbf{h}_{m,i,t}^{\mathsf H}\big)^{\!-1}\mathbf{h}_{m,k,t}}
{\!\big|\big(\mathbf{I}_{N_t}\!\!+\!\!\sum\limits_{i=1}^{K}\frac{\eta_{i,t}\alpha_i}{\sigma_{\mathrm{d}}^2}
\mathbf{h}_{m,i,t}\mathbf{h}_{m,i,t}^{\mathsf H}
\big)^{\!-1}
\mathbf{h}_{m,k,t}
\big|
},
\label{eq:optimal beamforming1}
\end{equation}
where \(p_{m,k,t}^{\mathrm d}\ge0\) denotes the DL transmit power allocated to
\(\mathrm{AU}_k\) on subchannel \(m\), satisfying $\sum_{m=1}^{M_{\mathrm{d}}}\sum_{k=1}^{K}p_{m,k,t}^{\mathrm d}
\leq P_{\max}^{\mathrm d}$. This completes the proof.

\bibliographystyle{IEEEtran}
\bibliography{my_ref}

\end{document}